\documentclass[journal,twoside,web]{ieeecolor}
\usepackage[dvipsnames]{xcolor}
\usepackage{generic}
\usepackage{cite}
\usepackage{amsmath,amssymb,amsfonts}
\usepackage{algorithmic}
\usepackage{dutchcal}
\usepackage{graphicx}
\usepackage[OT2,T1]{fontenc}
\usepackage{MnSymbol}
\usepackage{lipsum}
\usepackage{bbm}
\usepackage{color}
\newtheorem{theorem}{Theorem}

\newtheorem{definition}{Definition}

\newtheorem{proposition}{Proposition}

\newtheorem{lemma}{Lemma}
\newtheorem{corollary}[lemma]{Corollary}
\newtheorem{remark}{Remark}
\DeclareSymbolFont{cyrletters}{OT2}{wncyr}{m}{n}
\DeclareMathSymbol{\Sha}{\mathalpha}{cyrletters}{"58}
\DeclareMathAlphabet{\mathscr}{OT1}{pzc}{m}{it}
\usepackage{textcomp}
\def\BibTeX{{\rm B\kern-.05em{\sc i\kern-.025em b}\kern-.08em
 T\kern-.1667em\lower.7ex\hbox{E}\kern-.125emX}}
\begin{document}

\title{Solving Infinite-Dimensional Harmonic Lyapunov and Riccati equations}
\author{Pierre Riedinger and Jamal Daafouz 
\thanks{Pierre Riedinger and Jamal Daafouz are with Universit\'e de Lorraine, CNRS-CRAN UMR 7039, 2, avenue de la for\^et de Haye, 54516 Vandoeuvre-l\`es-Nancy Cedex, France.}
\thanks{This work is supported by HANDY project ANR-18-CE40-0010-02.} 
}
\maketitle
\begin{abstract} 
In this paper, we address the problem of solving infinite-dimensional harmonic algebraic Lyapunov and Riccati equations up to an arbitrary small error. This question is of major practical importance for analysis and stabilization of periodic systems including tracking of periodic trajectories. We first give a closed form of {a} Floquet factorization in the general setting of $L^2$ matrix functions and study the spectral properties of infinite-dimensional harmonic matrices and their truncated version. This spectral study allows us to propose a generic and numerically efficient algorithm to solve infinite-dimensional harmonic algebraic Lyapunov equations up to an arbitrary small error. We combine this algorithm with the Kleinman algorithm to solve infinite-dimensional harmonic Riccati equations and we apply the proposed results to the design of a harmonic LQ control with periodic trajectory tracking. 
\end{abstract}
\section{Introduction}
 Harmonic modelling and control is a topic of theoretical and practical interest for many application domains such as energy management (including AC-DC and AC-AC power converters) or embedded systems to mention few. In a recent paper \cite{Blin}, a complete and rigorous mathematical framework for harmonic modelling and control has been proposed. Basically, the harmonic modelling of a periodic system leads to an equivalent time invariant model but of infinite dimension. The states of this model (also called phasors) are the coefficients obtained using a sliding Fourier decomposition. One of the main results of \cite{Blin} establishes a strict equivalence between these two models and provides tools that allow to reconstruct time trajectories from harmonic ones. In this framework, the analysis and harmonic control design are considerably simplified as any available method for time-invariant systems can be a priori applied. For example, algebraic Lyapunov and Riccati equations \cite{Zhou2008,Blin} can be used to design a periodic state feedback for linear time periodic {(LTP)} systems. 
 
The main difficulty in applying available time invariant techniques is related to the infinite dimension nature of the obtained harmonic time invariant model. The associated infinite-dimensional harmonic state matrix is formed by the sum of a (block) Toeplitz matrix and a diagonal matrix. There is a huge literature concerning the study of infinite-dimensional Toeplitz matrices \cite{Beam,Bini,Iavernaro,Reichel92,Robol,gut,schmidt,Gohberg}. However, only few results concern harmonic state space matrices \cite{Zhou2004,Zhou,Kabamba,Farkas,Wereley_1990,Bolzern}. Most of the results dedicated to harmonic dynamical systems are based on {a} Floquet Factorization \cite{Floquet,Sinha,Wereley_1990,Montagnier, Zhou2004}. Floquet Factorization is an existence result and, as such, it is not constructive  \cite{Farkas,Floquet}. This is why  Floquet Factorization based methods are mainly dedicated to analysis and it is very difficult to extend them to control design. { As a consequence, solving infinite-dimensional harmonic algebraic Lyapunov and Riccati equations is a very challenging problem \cite{Zhou,Zhou2011}. 
	
%

The main objective of our paper is to propose efficient algorithms with low computational burden and of interest for both analysis and control design.  We first provide a simple and closed form formula to determine {a} Floquet factorization in the general case of periodic and $L^2$ matrix functions.  {As the proposed Floquet Factorization leads to a Jordan normal form representation of the harmonic state operator, it also solves the associated eigenvalue problem.} 
This new result allows us to perform a detailed spectral analysis of the harmonic state space matrix and its truncated version.  In particular, it is shown that the harmonic state space matrix  is an unbounded operator on $\ell^2$ with a discrete spectrum and that a truncation of a Hurwitz harmonic matrix may never be Hurwitz, regardless of the truncation order. To our knowledge, this is the first time that such a phenomenon is highlighted. It has a major impact in deriving tools for analysis and harmonic control design. As a consequence, if we consider a Hurwitz infinite-dimensional harmonic matrix, there may be no positive definite solution to the associated truncated harmonic Lyapunov equation whatever the considered truncation order.
To overcome this difficulty, we use Lyapunov symbolic equations associated to harmonic Lyapunov equations to provide an efficient algorithm that allows to recover the solution of the infinite-dimensional harmonic Lyapunov equation up to an arbitrarily small error. We extend this result to solve infinite-dimensional harmonic Riccati equations and provide a Kleinman's like algorithm  \cite{Kleinman} in the harmonic framework.  {This generic result is made possible by the fact that we do not use a Floquet factorization to solve a harmonic Lyapunov equation at each step.} To demonstrate that our results can be used for control design, we treat the problem of periodic trajectories tracking using a harmonic linear quadratic control as an illustrative example.

The paper is organized as follows. We first give some mathematical preliminaries in the next section before stating in section III the problem we are interested in. 
In section IV, we provide a complete and simple characterization of a Floquet factorization in the general case of $L^2$ matrix functions and analyze the spectral properties of the harmonic state space operator and its truncated version. The main contribution of our paper is detailed in section V where an efficient algorithm to solve up to an arbitrarily small error infinite dimensional harmonic Lyapunov equations is derived. This algorithm is extended to infinite dimensional harmonic Riccati equations in section VI. We illustrate the results of this paper in section VII where a design of an harmonic LQ control for a 2-dimensional LTP system is proposed. Section VIII is dedicated to the conclusions.

{\bf Notations: }The transpose of a matrix $A$ is denoted $A'$ and $A^*$ denotes the complex conjugate transpose $A^*=\bar A'$. The $n$-dimensional identity matrix is denoted $Id_n$. The infinite identity matrix is denoted $\mathcal{I}$. The $n\times n$ matrix of ones is denoted $\mathbbm{1}_{n,n}$. The flip matrix $J_m$ is the $(2m+1) \times (2m+1)$ matrix having 1 on the anti-diagonal and zeros elsewhere. The product $\cdot$ refers to the Hadamard product (known also as element-by-element multiplication). $A \otimes B$ is the Kronecker product of two matrices $A$ and $B$. $L^{p}$ (resp. $\ell^p$) denotes the Lebesgues spaces of $p-$integrable functions (resp. $p-$summable sequences) for $1\leq p\leq\infty$. $L_{loc}^{p}$ is the set of locally $p-$integrable functions i.e. on any compact set. The notation $f(t)=g(t)\ a.e.$ means almost everywhere in $t$ or for almost every $t$. We denote by $col(X)$ the vectorization of a matrix $X$, formed by stacking the columns of $X$ into a single column vector. 	We use $\sigma^+$ to denote the largest singular value. 
To simplify the notations, $L^p([a,b])$ or $L^p$ will be often used instead of $L^p([a,b],\mathbb{C}^n)$. 
For example, $x\in L^2([a,b])$ means $x \in L^2([a,b],\mathbb{C}^n)$.
\color{black}

\section{Mathematical Preliminaries}

We first start be recalling the definition of the sliding Fourier decomposition over a 
window of length $T$ and the so-called "Coincidence Condition" introduced in \cite{Blin}. 

\begin{definition}The sliding Fourier decomposition over a window of length $T$ from $ L^{2}_{loc}(\mathbb{R},\mathbb{C}^n)$ to $L^{\infty}_{loc}(\mathbb{R},\ell^2(\mathbb{C}^n))$ is defined by:
	$$X:=\mathcal{F}(x)$$
	where the time-varying infinite sequence $X$ is defined by:
	\begin{equation*}
		t\mapsto X(t):=(\mathcal{F}(x_1)(t), \cdots,\mathcal{F}(x_n)(t))\end{equation*}
	and where for $i:=1,\cdots,n$, the vector $\mathcal{F}(x_i):=(\cdots, X_{i,-1}, X_{i,0}, X_{i,1},\cdots)$, has infinite components $X_{i,k} $, $k\in \mathbb{Z}$ satisfying:
	$X_{i,k}(t):=\frac{1}{T}\int_{t-T}^t x_i(\tau)e^{-j\omega k \tau}d\tau.
	$
	The vector $X_k:=(X_{1,k}, \cdots, X_{n,k})$ is called the $k-$th phasor of $X$. \end{definition}
	\begin{definition}\label{H} We say that $X$ belongs to $H$ if $X$ is an absolutely continuous function (i.e $X\in C^a(\mathbb{R},\ell^2(\mathbb{C}^n))$ and fulfills for any $k$ the following condition: \begin{equation*}\dot X_k(t)=\dot X_0(t)e^{-j\omega k t} \ a.e.\end{equation*}
\end{definition}
Similarly to the Riesz-Fisher theorem which establishes a one-to-one correspondence between the spaces $L^2$ and $\ell^2$, the following "Coincidence Condition" establishes a one-to-one correspondence between the space $L_{loc}^2$ and the space $H$.
\begin{theorem}[Coincidence Condition \cite{Blin}]\label{coincidence}For a given $X\in L_{loc}^{\infty}(\mathbb{R},\ell^2(\mathbb{C}^n))$, there exists a representative $x\in L^2_{loc}(\mathbb{R},\mathbb{C}^n)$ of $X$, i.e. $X=\mathcal{F}(x)$, if and only if $X$ belongs to $H$.
\end{theorem}

In the sequel, we provide some mathematical preliminaries related to block Toeplitz matrices and operator norms. These preliminaries are adaptations to our setting of some mathematical results borrowed from \cite{Bini,Massei,Bottcher,Bini2,Iavernaro,Reichel92,Robol,gut}. 

\subsection{ Finite and infinite Toeplitz and block Toeplitz matrices}
Consider a $T-$periodic $L^2([0 \ T],\mathbb{C})$ signal $a$, its associated Toeplitz matrix $\mathcal{T}(a)$  
$$\mathcal{T}(a):= (t_{ij}), {i,j}\in \mathbb{Z} \text{ such that }t_{ij} := a_{i-j}$$
and its symbol (Laurent series) $a(z)=\sum_{k=-\infty}^{+\infty}a_kz^k$ where $a_k$, $k\in \mathbb{Z}$, are the phasors of $a(\cdot)$. Define the semi-infinite Toeplitz matrix $$\mathcal{T}_s(a):= (t_{ij}), {i,j}\in \mathbb{Z}^+ \text{ such that }t_{ij} := a_{i-j}$$
and let $a^+(z):=\sum_{k>0}a_kz^k$ and $a^-(z):=\sum_{k>0}a_{-k}z^{-k}$. We associate with $a^+(z)$ and $a^-(z)$ the following semi-infinite Hankel matrices
\begin{align*}
\mathcal{H}(a^+) &:= (h^+_{ij}), {i,j}\in \mathbb{Z}^{+*},\ h^+_{ij} := a_{i+j-1},\\
\mathcal{H}(a^-) &:= (h^-_{ij}), {i,j}\in \mathbb{Z}^{+*},\ h^-_{ij} := a_{-i-j+1}\end{align*}
 
Given a symbol $a(z)$ and $m \in \mathbb{Z}^+$, we denote by $\mathcal{T}_m(a)$, the $(2m+1) \times (2m+1)$ leading principal submatrices of $\mathcal{T}(a)$. 
 We denote also by $\mathcal{H}_{(p,q)}(a)$, for $p,q>0$, the $(2p+1)\times(2q+1)$ Hankel matrix obtained selecting the first $(2p+1)$ rows and $(2q+1)$ columns of $\mathcal{H}(a)$. For clarity purpose, we provide in Fig.~\ref{fig20} a block decomposition of an infinite Toeplitz matrix $\mathcal{T}(a)$ to illustrate how the matrices defined above appear. This block decomposition will be useful in the sequel.
 
\begin{figure}[h]\begin{center}
\includegraphics[width=\linewidth,height=6.5cm]{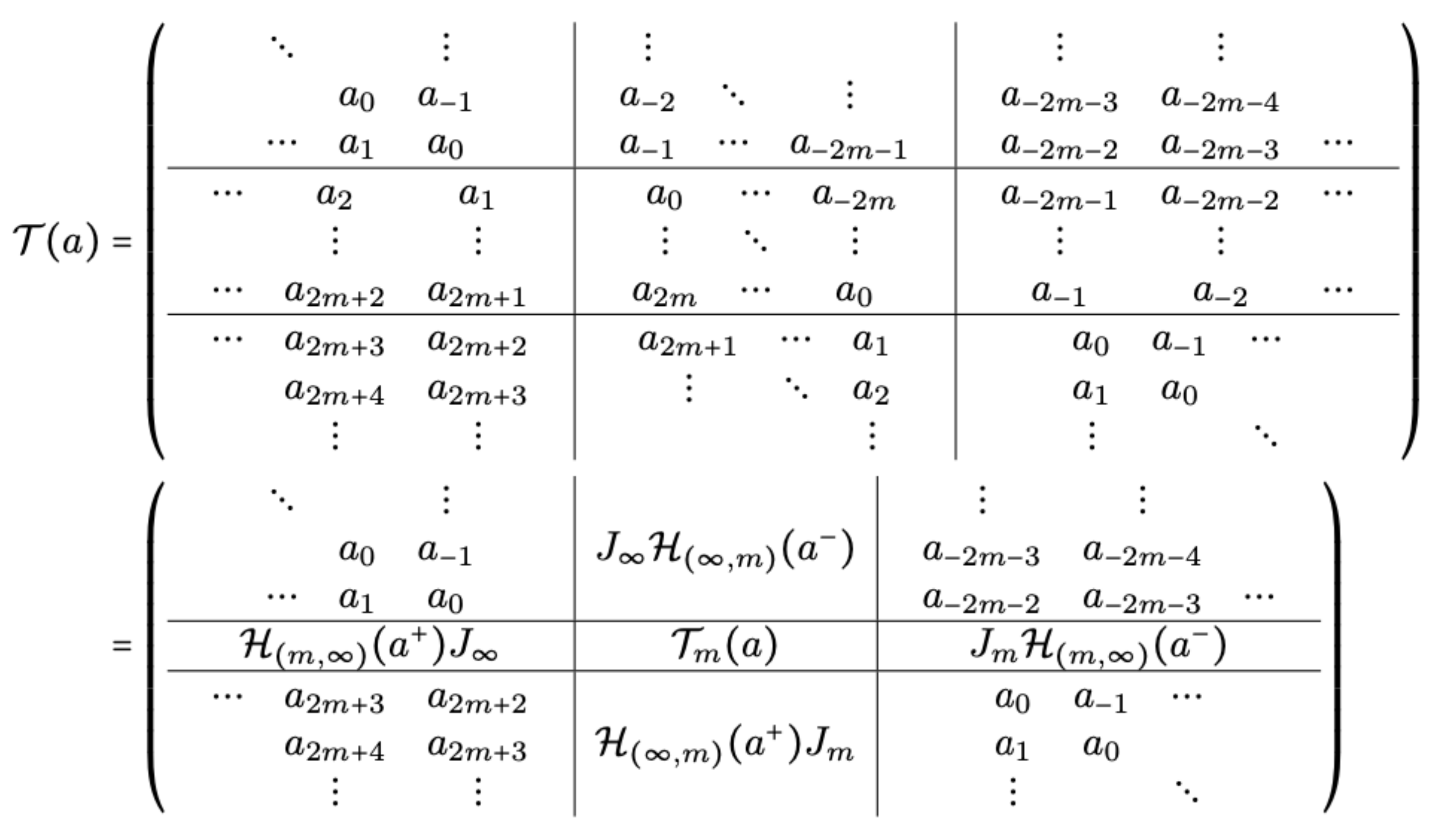}
\caption{Block decomposition of an infinite Toeplitz matrix $\mathcal{T}(a)$} \label{fig20}
\end{center}
\end{figure}

\begin{definition}\label{bt}
	The block Toeplitz transformation of a $T-$periodic $n\times n$ matrix function $A\in L^{2}([0 \ T],\mathbb{C}^{n \times n})$, denoted $\mathcal{A}:=\mathcal{T}(A)$, defines a constant $n\times n$ block Toeplitz and infinite-dimensional matrix as follows: 
	$$\mathcal{A}:=\left(\begin{array}{cccc}
		\mathcal{A}_{11} & \mathcal{A}_{12} & \cdots & \mathcal{A}_{1n} \\
		\mathcal{A}_{21} & \mathcal{A}_{22} & & \vdots \\
		\vdots & & \ddots& \vdots \\
		\mathcal{A}_{n1} & \cdots & \cdots & \mathcal{A}_{nn}\end{array}\right)$$ where the infinite matrices $\mathcal{A}_{ij}:=\mathcal{T}(a_{ij})$, $i,j:=1,\cdots,n$, are the Toeplitz transformations of the entries $a_{ij}(t)$ of the matrix $A(t)$:
	\begin{align*}
		\mathcal{T}(a_{ij}):=
		\left[
		\begin{array}{ccccc}
			\ddots & & \vdots & &\udots \\ & a_{ij,0} & a_{ij,-1} & a_{ij,-2} & \\
			\cdots & a_{ij,1} & a_{ij,0} & a_{ij,-1} & \cdots \\
			& a_{ij,2} & a_{ij,1} & a_{ij,0} & \\
			\udots & & \vdots & & \ddots\end{array}\right],\end{align*}
	with $a_{ij,k} :=\frac{1}{T}\int_{t-T}^t a_{ij}(\tau)e^{-j\omega k \tau}d\tau$. \\
\end{definition}

{
In the sequel, to avoid confusions, for any $T-$periodic matrix function $A\in L^2$, we denote by  ${\bf A}:=\mathcal{F}(A)$ its Fourier decomposition and by  $\mathcal{A}:=\mathcal{T}(A)$ its Toeplitz transformation. 
The $m-$truncation of the $n \times n$ block Toeplitz matrix  $\mathcal{A}$ is defined by the $m-$truncation $\mathcal{T}(a_{ij})_m$ of all its entries $(i,j)$.}
The symbol matrix $A(z)$ associated to a $n\times n$ block Toeplitz matrix is given by:
\begin{equation}A(z):=\left(\begin{array}{cccc}
a_{11}(z) & a_{12}(z) & \cdots & a_{1n}(z) \\
a_{21}(z) & a_{22}(z) & & \vdots \\
\vdots & & \ddots& \vdots \\
a_{n1}(z) & \cdots & & a_{nn}(z)\end{array}\right).\label{symA}
\end{equation}
The $n\times n$ block Hankel matrices $\mathcal{H}(A^+)$, $\mathcal{H}(A^-)$ are also defined respectively by $\mathcal{H}(A^+)_{ij}:=\mathcal{H}(a^+_{ij})$ and $\mathcal{H}(A^-)_{ij}:=\mathcal{H}(a^-_{ij})$ for $i,j:=1,\cdots,n$. {In the same way, their subprincipal submatrices $\mathcal{H}(A^+)_{(p,q)}$, $\mathcal{H}(A^-)_{(p,q)}$ for $p,q >0$ are obtained by considering the subprincipal submatrices of the entries $\mathcal{H}(a^+_{ij})_{(p,q)}$ and $\mathcal{H}(a^-_{ij})_{(p,q)}$ for $i,j:=1,\cdots,n$.}
 
\begin{theorem} \label{product}Let $A(z)$, $B(z)$ be two symbol matrices and $C(z) :=A(z)B(z)$. Then,
\begin{equation}
\mathcal{T}_s(A)\mathcal{T}_s(B) = \mathcal{T}_s(C) - \mathcal{H}(A^+)\mathcal{H}(B^-) \label{ee1}\end{equation} and
\begin{align}
\mathcal{T}_m(A)\mathcal{T}_m(B) &= \mathcal{T}_m(C)- \mathcal{H}_{(m,\eta)}(A^+)\mathcal{H}_{(\eta,m)}(B^-)\nonumber \\&- \mathcal{J}_{n,m}\mathcal{H}_{(m,\eta)}(A^-)\mathcal{H}_{(\eta,m)}(B^+) \mathcal{J}_{n,m}, \label{ee2}\end{align}
where $\mathcal{J}_{n,m}:=Id_n\otimes J_m$ and $\eta$ is such that $2\eta\geq \min d^o (A(z),B(z))$

\end{theorem}
\begin{proof} A classical result states that the product of two infinite Toeplitz matrices is a Toeplitz matrix. This means that for two symbols $a(z)$ and $b(z)$ with $c(z):=a(z)b(z)$, we have $\mathcal{T}(c)=\mathcal{T}(a)\mathcal{T}(b)$. The formula $\mathcal{T}_m(a)\mathcal{T}_m(b) = \mathcal{T}_m(c) - \mathcal{H}_{(m,\eta)}(a^+)\mathcal{H}_{(\eta,m)}(b^-) -J_m\mathcal{H}_{(m,\eta)}(a^-)\mathcal{H}_{(\eta,m)}(b^+)J_m,$ where $2\eta\geq\min(d^o(a(z),d^o(b(z))$, is directly obtained by applying to each Toeplitz matrix $\mathcal{T}(a)$, $\mathcal{T}(b)$ and $\mathcal{T}(c)$ the block decomposition of Fig.~\ref{fig20}. If the degrees of $a(z)$ and $b(z)$ are unknown or infinite then $\eta$ can be set to $+\infty$. 
For the $n\times n$ block Toeplitz case, the result is obtained by considering the entries $(i,j)$, $i,j =1,\cdots,n$ of the block matrix product:
\begin{align*}
(\mathcal{T}_m(A)\mathcal{T}_m(B))_{ij}&=\sum_{k=1}^n\mathcal{T}_m(a_{ik})\mathcal{T}_m(b_{kj})
\end{align*}
and decomposing each term of the sum, that is:
\begin{align*}
(\mathcal{T}_m(A)\mathcal{T}_m(B))_{ij}=\sum_{k=1}^n&(\mathcal{T}_m(c_{ijk})-\mathcal{H}_{(m,\eta)}(a_{ik}^+)\mathcal{H}_{(\eta,m)}(b_{kj}^-)\nonumber\\ - J_m&(\mathcal{H}_{(m,\eta)}(a_{ik}^-)\mathcal{H}_{(\eta,m)}(b_{kj}^+) J_\eta)
\end{align*}
where $c_{ijk}(z):=a_{ik}(z)b_{kj}(z)$. The results  follows for \eqref{ee2} and also for \eqref{ee1} using similar steps from the symbol formula $\mathcal{T}_s(a)\mathcal{T}_s(b) = \mathcal{T}_s(c) - \mathcal{H}(a^+)\mathcal{H}(b^-)$ (see \cite{Bottcher}). \end{proof}
\begin{figure}\begin{center}
\includegraphics[scale=0.25]{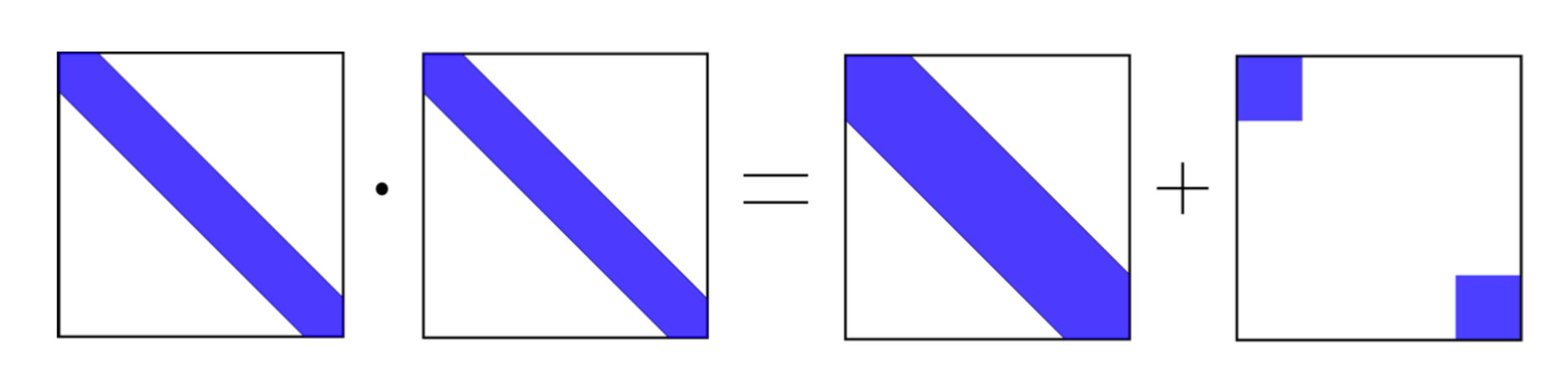}
\caption{Multiplication of two finite dimensional banded Toeplitz matrices}\label{fig1}
\end{center}
\end{figure}
An illustration of the above theorem is given in Fig.~\ref{fig1} for $n:=1$ with $a(z)$ and
$b(z)$ Laurent polynomials of degree much less than $m$ so that $\mathcal{T}_m(a)$ and $\mathcal{T}_m(b)$
are banded.
If $a(z) :=\sum^k_{i=-k} a_iz^i$ and $b(z) :=\sum^k_{i=-k} b_iz^i$ with $k$ much smaller than $m$, then
the matrices $E^+:=\mathcal{H}_m(a^+)\mathcal{H}_m(b^-)$ and $E^-:= J_m\mathcal{H}_m(a^-)\mathcal{H}_m(b^+) J_m$ have disjoint supports located in the upper leftmost corner and in the lower rightmost corner, respectively. As a consequence,
$\mathcal{T}_m(a)\mathcal{T}_m(b)$ can be represented as the sum of the Toeplitz matrix associated with $c(z)$ and two correcting terms $E^+$ and $E^-$.

We end these preliminaries on block Toeplitz matrices by defining what we call letf and right truncations and two results given without proofs as they follow from the block decomposition of Fig.~\ref{fig20}.
\begin{definition}\label{truncdef} The left $m-$truncation (resp. right $m-$truncation) of a $n\times n$ block Toeplitz infinite matrix $\mathcal{A}$ is given by:
	{\small$$\mathcal{A}_{m^+}:=\left(\begin{array}{cccc}
			\mathcal{A}_{{11}_{m^+}} & \mathcal{A}_{{12}_{m^+}} & \cdots & \mathcal{A}_{{1n}_{m^+}} \\
			\mathcal{A}_{{21}_{m^+}} & \mathcal{A}_{{22}_{m^+}} & & \vdots \\
			\vdots & & \ddots& \vdots \\
			\mathcal{A}_{{n1}_{m^+}} & \cdots & \mathcal{A}_{{n(n-1)}_{m^+}} & \mathcal{A}_{{nn}_{m^+}}\end{array}\right)$$} 
	(resp. $\mathcal{A}_{m^-}$) where $\mathcal{A}_{{ij}_{m^+}}$, $i,j:=1,\cdots, n$
	are obtained by suppressing in the infinite matrices $\mathcal{A}_{ij}$ all the columns and lines having an index strictly smaller than $-m$ (respectively strictly greater than $m$).  {Finally, the $m-$truncation is obtained by applying successively a left and a right $m-$truncations.}
\end{definition}
\begin{proposition}\label{product3}Let $a(z)$ be a symbol and $x:=(x_k)_{k\in \mathbb{Z}}$ an infinite vector of complex numbers.
Define the $m-$truncation of $x$ by $x|_m:=(x_{-m},\cdots, x_{m})$ and consider the semi-infinite vectors $x|_m^+:=(x_{m+1},x_{m+2}, \cdots)$ and $x|_m^-:=(\cdots, x_{-m-2},x_{-m-1})$. Let $\breve x$ be the infinite vector given by $\breve x:=(\cdots, 0,x|_m,0,\cdots)$. Then, the following relations hold true:
\begin{align}
\mathcal{T}(a)\breve x=\left[\begin{array}{c}
J_\infty\mathcal{H}_{(\infty,m)}(a^-) \\
\mathcal{T}_m(a) \\
\mathcal{H}_{(\infty,m)}(a^+)J_m\end{array}\right]x|_m\label{z1}\end{align}
\begin{align}
\mathcal{T}_m(a)x|_m=(\mathcal{T}(a)x)|_m&-\mathcal{H}_{(m,\infty)}(a^+)J_\infty x|_m^-\nonumber\\&-J_m \mathcal{H}_{(m,\infty)}(a^-)x|_m^+ \label{z2}\end{align}
 \end{proposition}
The next proposition is a generalization of Proposition~\ref{product3} to the case of $n\times n$ block Toeplitz matrices.
\begin{proposition}\label{general}Let $A(z)$ be a $n \times n$ symbol matrix and $x:=(x_1,\cdots, x_n)$ a vector whose components $x_i$ are infinite sequences $x_i:=(\cdots, x_{i,-1},x_{i,0},x_{i,1},\cdots)$.
Define $x|_m:=(x_1|_m,\cdots,x_n|_m)$ the $m-$truncation of $x$ where for $i :=1,\cdots ,n$, $x_i|_m:=(x_{i,-m},\cdots, x_{i,m})$.
Define also the semi infinite vectors $x_i|_m^+:=(x_{i,m+1},x_{i,m+2}, \cdots)$ and $x_i|_m^-:=(\cdots, x_{i,-m-2},x_{i,-m-1})$. 
Set $\breve x:=(\breve x_1,\breve x_2, \cdots,\breve x_n)$ with $\breve{x}_i:=(\cdots, 0,x_i|_m,0,\cdots)$ for any $i :=1,\cdots ,n$.
Then, we have:
$$(\mathcal{T}(A)\breve x)_i=\sum_{j=1}^n\mathcal{T}(a_{ij})\breve x_j$$
where $\mathcal{T}(a_{ij})\breve x_j $ is given by \eqref{z1}
and
\begin{align*}
\mathcal{T}_m(A)x|_m&=\sum_{j=1}^n(\mathcal{T}(a_{ij})x_j)|_m\\&-\mathcal{H}_{(m,\infty)}(a_{ij}^+)J_\infty x_j|_m^--J_m\mathcal{H}_{(m,\infty)}(a_{ij}^-)x_j|_m^+
\end{align*}
\end{proposition}
\subsection{ Operator norms}
We provide here some results concerning operator norms to be used in the sequel. Recall that the norm of an operator $M$ from $\ell^p$ to $\ell^q$ is given by 
$$\|M\|_{\ell^p,\ell^q}:=\sup_{\|X\|_{\ell^p}=1}\| MX \|_{\ell^q}.$$
This operator norm is sub-multiplicative i.e. if $M: \ell^p \rightarrow \ell^q$ and $N: \ell^q \rightarrow \ell^r$ then
$\|NM\|_{\ell^p,\ell^r} \leq \|M\|_{\ell^p,\ell^q} \|N\|_{\ell^q,\ell^r}$.
If $p=q$, we use the notation: $\|M\|_{\ell^p}:=\|M\|_{\ell^p,\ell^p}$.
\begin{definition}Consider a vector $x(t)\in L^2([0 \ T],\mathbb{C}^n)$ and define $X:=\mathcal{F}(x)$ with its symbol $X(z)$.
	The $\ell^2 -$norm of $X(z)$ is given by:
	$$\|X(z)\|_{\ell^2}:=\|X\|_{\ell^2}$$ where $\|X\|_{\ell^2}:=\left(\sum_{k\in\mathbb{Z}}|X_k|^2\right)^{\frac{1}{2}}$.
\end{definition}

\begin{theorem}\label{borne} Let $A(t)\in L^2([0 \ T],\mathbb{C}^{n\times m})$. Then, $\mathcal{A}:=\mathcal{T}(A)$ is a bounded operator on $\ell^2$ if and only if $A\in L^{\infty}([0\ T],\mathbb{C}^{n\times m} )$.
	Moreover, we have:
	\begin{enumerate}
		\item the operator norm induced by the $\ell^2$-norm satisfies: $$\|A(z)\|_{\ell^2}=\|\mathcal{A}\|_{\ell^2}=\|A\|_{L^{\infty}}$$
		\item the operator norm of the semi infinite Toeplitz matrix satisfies: $\|\mathcal{T}_s(A)\|_{\ell^2}=\|\mathcal{A}\|_{\ell^2}$
		\item the operator norm of the Hankel operators $\mathcal{H}(A^+)$, $\mathcal{H}(A^-)$ satisfies: 
		$\|\mathcal{H}(A^-)\|_{\ell^2}\leq \|A\|_{L^{\infty}}$ and 
		$\|\mathcal{H}(A^+)\|_{\ell^2}\leq \|A\|_{L^{\infty}}$
		\item the operator norm related to the left and right $m-$truncations satisfies:	$\|\mathcal{A}_{m^+}\|_{\ell^2}=\|\mathcal{A}_{m^-}\|_{\ell^2}=\|\mathcal{A}\|_{\ell^2}=\|A(t)\|_{L^\infty}$
	\end{enumerate}
\end{theorem}
\begin{proof}See Part V p.p. 562-574 of \cite{Gohberg}.
\end{proof}
\begin{proposition}\label{fro} Let $P(\cdot)$ be a matrix function in $ L^\infty([0 \ T],\mathbb{C}^{n\times n})$. Define ${\bf P}:=\mathcal{F}(P)$ and $\mathcal{P}:=\mathcal{T}(P)$. If 
	$\|col({\bf P})\|_{\ell^2}\leq \epsilon$ 
	then $\|\mathcal{P}\|_{\ell^2}\leq \epsilon$.
\end{proposition}
\begin{proof}
	Using Riesz-Fisher Theorem, we have:
	\begin{align*}
		\|col({\bf P})\|_{\ell^2}&=\|col( P)\|_{L^2}=(\sum_{i,j=1}^n\|P_{ij}\|^2_{L^2})^{1/2}=\|P\|_F
	\end{align*}
	where $\|P(t)\|_F$ stands for the Frobenius norm.
	As $P \in L^\infty([0 \ T],\mathbb{C}^{n\times n})$, H$\ddot{\mbox{o}}$lder's inequality implies $Px\in L^2([0 \ T],\mathbb{C}^{n})$ for any $x\in L^2([0 \ T],\mathbb{C}^{n})$. Thus, the result follows from the following relations between operator norms: 
	\begin{align*}
			\|\mathcal{P}\|_{\ell^2}&=\sup_{\|X\|_{\ell^2}=1}(<\mathcal{P}X,\mathcal{P}X>_{\ell^2})^{1/2}\\
			&=\sup_{\|x\|_{L^2}=1}(<Px,Px)>_{L^2})^{1/2}\\
			&\leq (trace( P^*P))^{1/2}= \|P\|_F
	\end{align*}
where $<\cdot,\cdot>$ stands for the scalar product. 
\end{proof}

\section{Problem statement}
To formulate the problem we are interested in, we need to recall some key results from \cite{Blin}. Under the "Coincidene Condition" of Theorem~\ref{coincidence}, it is established in \cite{Blin} that any periodic system having solutions in Carath\'eodory sense can be transformed by a sliding Fourier decomposition into a time invariant system. For instance, consider $T-$periodic functions $A(\cdot)$ and $B(\cdot)$ respectively of class $L^2([0\ T],\mathbb{R}^{n\times n})$ and $L^{\infty}([0\ T],\mathbb{R}^{n\times m_u})$ and let the linear time periodic system:
\begin{align}\dot x(t)=A(t)x(t)+B(t)u(t)\quad x(0):=x_0\label{ltp}\end{align}
If, $x$ is a solution associated to the control $u\in L_{loc}^2(\mathbb{R},{\mathbb{R}^{m_u})}$ of the linear time periodic system (\ref{ltp}) then, $X:=\mathcal{F}(x)$ is a solution of the linear time invariant system:
\begin{align}
	\dot X(t)=(\mathcal{A}-\mathcal{N})X(t)+\mathcal{B}U(t), \quad X(0):=\mathcal{F}(x)(0) \label{ltih}
\end{align}
where $\mathcal{A}:=\mathcal{T}(A)$, $\mathcal{B}:=\mathcal{T}(B)$ and 
\begin{equation}\mathcal{N}:=Id_n\otimes diag(j\omega k,\ k\in \mathbb{Z})\label{N}\end{equation}
Reciprocally, if $X\in H$ is a solution of \eqref{ltih} with $U\in H$, then its representative $x$ 
(i.e. $X=\mathcal{F}(x))$ is a solution of \eqref{ltp}. Moreover, for any $k\in\mathbb{Z}$, the phasors $X_k \in C^1(\mathbb{R},\mathbb{C}^n)$ and $\dot X\in C^0(\mathbb{R},\ell^{\infty}(\mathbb{C}^n))$. As the solution $x$ is unique for the initial condition $x_0$, $X$ is also unique for the initial condition $X(0):=\mathcal{F}(x)(0)$. In addition, it is proved in \cite{Blin} that one can reconstuct time trajectories from harmonic ones, that is:
	\begin{align}\label{recos} x(t)&=\mathcal{F}^{-1}(X)(t):=\sum_{p=-\infty}^{+\infty} X_p(t)e^{j\omega p t}+\frac{T}{2}\dot X_0(t)\end{align}
	where $X_{k}=(X_{1,k}, \cdots, X_{n,k})$ for any $k\in \mathbb{Z}$.

In the same way, a strict equivalence between a periodic differential Lyapunov equation and its associated harmonic algebraic Lyapunov equation is also proved \cite{Blin}. Namely, let $Q\in L^{\infty}([0\ T])$ be a $T$-periodic symmetric 
	and positive definite matrix function. $P$ is the unique $T$-periodic symmetric positive definite solution of the periodic differential Lyapunov equation:
	\begin{equation*}\dot P(t)+A'(t)P(t)+P(t)A(t)+Q(t)=0,\end{equation*}
	if and only if $\mathcal{P}:=\mathcal{T}(P)$ is the unique hermitian and positive definite solution of the harmonic algebraic Lyapunov equation:
	\begin{equation}
		\mathcal{P}(\mathcal{A}-\mathcal{N})+(\mathcal{A}-\mathcal{N})^*\mathcal{P}+\mathcal{Q}=0,\label{al}
	\end{equation}
	where $\mathcal{Q}:=\mathcal{T}(Q)$ is hermitian positive definite and $\mathcal{A}:=\mathcal{T}(A)$.
	Moreover, $\mathcal{P}$ is a bounded operator on $\ell^2$ and $P$ is an absolutely continuous function.


These results are of great interest. Solving an algebraic Lyapunov equation rather than a periodic differential Lyapunov equation is worthwile for analysis and control design provided coping with the infinite dimension nature of equation~\eqref{al}. The main difficulty is related to the diagonal matrix $\mathcal{N}$ defined by \eqref{N} which is not a Toeplitz matrix nor a compact operator. Hence, the harmonic algebraic Lyapunov equation~\eqref{al} cannot be expressed as a simple product of symbols as in the classical Toeplitz case \cite{Robol}. In \cite{Zhou}, \cite{Zhou2011}, the authors propose to use {a} Floquet factorization but the determination of this Floquet factorization is not so simple \cite{Sinha,Kabamba,Zhou2}. Furthermore, for control design purpose, it would not be appropriate to proceed this way since the input matrix remains a full matrix with no particular and usefull structure in the harmonic domain. 

The main objective of our paper is to show how the solution of the infinite-dimensional HLE \eqref{al} can be obtained from a finite dimensional problem up to an arbitrary error. As we will see, this is a practical result that avoids the computation of {a} Floquet factorization and reduces significantly the computation burden. We also extend our result to harmonic Riccati equations encountered in periodic optimal control. To this end, the characterization of the spectrum of the harmonic state operator $(\mathcal{A}-\mathcal{N})$ is of major importance and plays a key role in the derivation of the main contributions of our paper.

 \section{Spectral properties of $(\mathcal{A}-\mathcal{N})$}
 In this section, we provide a simple closed form formula for {a} Floquet factorization, characterize the spectrum of the harmonic state operator $(\mathcal{A}-\mathcal{N})$ and study the spectral properties of its truncated version. As noticed before, the harmonic state matrix $(\mathcal{A}-\mathcal{N})$ is not Toeplitz because of $\mathcal{N}$. This term has an important impact on the spectral properties of $(\mathcal{A}-\mathcal{N})$. For instance, we know that the spectrum of a Toeplitz matrix $\mathcal{A}$ is continuous \cite{Reichel92,Beam,schmidt} and bounded when $A(z)$ belongs to $\ell^2$. However, we will see in the sequel that the spectrum of $(\mathcal{A}-\mathcal{N})$ is unbounded and discrete. We will also explain how this spectrum behaves when applying a $m-$truncation $(\mathcal{A}_m-\mathcal{N}_m)$. 
 \subsection{A closed form formula for {a} Floquet factorization and spectral properties of $(\mathcal{A}-\mathcal{N})$}
Recall that the Floquet theorem \cite{Farkas,Zhou} states that for dynamical systems 
\begin{equation}\dot x(t)=A(t) x(t)  \label{hs2}\end{equation} 
with $A(t)$ piecewise continuous and $T-$periodic, the state transition matrix $\Phi(t, 0)$ has {a} Floquet factorization $\Phi(t, 0)= W(t)e^{Qt}$, where $Q$ is a constant matrix and $W(t)$ is continuous in $t$, nonsingular and $T-$periodic in $t$. Moreover, the state transformation $z(t):=W(t)^{-1}x(t)$ leads to a LTI system:
$$\dot z(t)=Qz(t)$$ 
and the harmonic system associated to \eqref{hs2}:
$$\dot X=(\mathcal{A}-\mathcal{N})x$$
becomes : 
$$\dot Z=(\mathcal{Q}-\mathcal{N})Z$$
with $Z:=\mathcal{F}(z)$ and $\mathcal{Q}:=\mathcal{I}\otimes Q$.
Unfortunately, this result is an existence result and, as such, it is not constructive.
	One may find algorithms to determine $W(t)$ and $Q$ as those proposed in \cite{Castelli} and \cite{Zhou2}. Here, we show that a more simple characterization of {a} Floquet factorization can be obtained with $Q$ in a Jordan normal form and $W(t)$ easily determined as the solution of an initial value problem with explicit initial conditions. Moreover, our result is given with the assumption that the $T-$periodic matrix function $A$ belongs to $L^2$ which is more general than existing results.
	
When $A\in L^2([0\ T],\mathbb{R}^{n\times n})$, the initial value problem defined by \eqref{hs2} and $x(0):=x_0$ admits an unique solution in the Carath\'eodory sense.
We can define $n$ linearly independent fundamental solutions denoted $x^{(i)}(t)$ having $e_i$ as initial conditions. As a consequence, the Wronski matrix \begin{equation} \Phi(t,0)=[x^{(1)}(t),\cdots, x^{(n)}(t)]\label{wron}\end{equation} is the state transition matrix and for any time $t$, $x(t)= \Phi(t,0)x_0$ is solution of the initial value problem. Moreover, $ \Phi(t,0)$ is non singular, absolutely continuous and therefore almost everywhere differentiable. This is important to characterize the eigenvalues and eigenvectors of the harmonic operator $(\mathcal{A}-\mathcal{N})$ as shown in the next Theorem for the case when $ \Phi(T,0)$ is non defective.

\begin{theorem}\label{nondef} Assume that the $T-$periodic function $A(t)$ belongs to $L^2([0 \ T])$ and that $\Phi(T,0)$ is non defective. Let $\mu $ and $\phi$ be respectively an eigenvalue and an associated eigenvector of $\Phi(T,0)$.  
Then, 
$\lambda$ and $V$ are an eigenvalue and an eigenvector of $(\mathcal{A}-\mathcal{N})$
\begin{equation*}(\mathcal{A}-\mathcal{N})V=\lambda V\end{equation*}
if and only if $v:=\mathcal{F}^{-1}(V)$ is a $T-$periodic solution in the Carath\'eodory sense of the initial value problem

	\begin{equation}\dot v(t)=(A(t)-\lambda Id_n)v(t)\quad v(0):=\phi\label{vp}\end{equation} 
	where $\lambda:=\frac{1}{T}\log(\mu)$ (not necessarily its principal value). 
\end{theorem}
\begin{proof}Applying Theorem~4 in \cite{Blin}, it follows that a solution of \eqref{vp} is a solution of 
	\begin{equation}\dot V=(\mathcal{A}-\mathcal{N}-\lambda \mathcal{I})V\label{hvp} \end{equation} 
	where $V=:\mathcal{F}(v)$ and reciprocally (provided $V$ is a trajectory of \eqref{hvp} that belongs to $H$, see Definition~\ref{H}). If $v$ is $T-$periodic then $\dot V=0$. Thus, $\lambda$ and $V$ are necessarily an eigenvalue and an eigenvector of $(\mathcal{A}-\mathcal{N})$. Reciprocally, if $\lambda$ and $V$ are an eigenvalue and an eigenvector of $(\mathcal{A}-\mathcal{N})$ this means that $\dot{V}=0$ in (\ref{hvp}). As $V$ is constant, it belongs trivially to $H$. Hence, $V$ admits an absolutely continuous and $T-$periodic representative $v$ that satisfies \eqref{vp} a.e.
	Now, consider an eigenvalue $\mu$ and an associated eigenvector $\phi$ of $ \Phi(T,0)$, then $ \Phi(T,0)\phi=\mu \phi$. Notice that $\mu$ cannot be equal to zero since $\Phi(T,0)$ is not singular.
Define $\lambda:= \frac{1}{T}\log(\mu)$ (not necessarily as the principal value of $\log(\mu)$), then we have:
 \begin{equation}
		\phi=R(T,0)\phi\label{per}\end{equation}
	with $R(T,0):=e^{-\lambda T} \Phi(T,0)$.
	Moreover, as $ \Phi(t,0)$ is a.e. differentiable, we can write:
	$$\dot \Phi(t,0)=A(t) \Phi(t,0)\ a.e.$$
	Let $R(t,0):=e^{-\lambda t} \Phi(t,0)$. We have: 
	\begin{align}
	\dot{R}(t,0)&=-\lambda e^{-\lambda t} \Phi(t,0)+e^{-\lambda t}\dot \Phi(t,0)\\
		&=(A(t)-\lambda Id_n)R(t,0) \ a.e.
	\end{align}
Hence, $R(t,0)$ is the state transition matrix of the linear system \eqref{vp}. We conclude from \eqref{per} that the solution of the initial value problem \eqref{vp} defined by such a $\lambda$ and $\phi$ is $T-$periodic. 
\end{proof}

To generalize this result to the case where $ \Phi(T,0)$ is defective, let us consider a Jordan normal form of the matrix $ \Phi(T,0)$ and assume that $(\mu_1,\cdots,\mu_n)$ and 
\begin{equation}P_\Phi:=[\phi_1,\cdots, \phi_n]\label{gen}
\end{equation} 
are respectively the eigenvalues and the matrix formed by the generalized eigenvectors of $ \Phi(T,0)$. To ease the presentation, we 
assume without loss of generality that 
$$P_\Phi^{-1} \Phi(T,0)P_\Phi=\left(\begin{array}{cccc}\mu & 1 & 0 & 0 \\0 & \mu & \ddots & 0 \\0 & 0 & \ddots & 1 \\0 & 0 & 0 & \mu\end{array}\right)$$
with $\phi_i\in \mathcal{K}er( \Phi(T,0)-\mu Id)^i$, for $i:=1,\cdots,n$. 

\begin{theorem} \label{def} Consider $\lambda:=\frac{1}{T}\log(\mu)$ (not necessarily the principal value) and let {$V_{0}:=0$ and $v_{0}:=0$.} 
For $i:=1,...,n$, $V_i$ is a generalized eigenvector associated to $\lambda$
		\begin{align}(\mathcal{A}-\mathcal{N})V_{i+1}&=\lambda V_{i+1}+\frac{1}{T\mu}V_i\label{vp4}\end{align} 
if and only if $v_i:=\mathcal{F}^{-1}(V_i)$ is a $T-$periodic solution in Carath\'eodory sense of the initial value problem: \begin{align}
		\dot v_i(t)&=(A(t)-\lambda Id_n)v_i(t)-\frac{1}{T\mu}v_{i-1}(t),\ v_i(0):=\phi_i\label{vp3}\end{align} 
	where $\phi_i$ are provided by \eqref{gen}.
\end{theorem}
\begin{proof}The strict equivalence between \eqref{vp4} and \eqref{vp3} is obtained following similar steps as in the proof of Theorem~\ref{nondef}.
For $i:=1$, as $\phi_1$ is an eigenvector of $\Phi(T,0)$, the result is already proved (Theorem~\ref{nondef}) and $v_{1}(t):=R(t,0)v_1(0)$ with
	$R(t,0):=e^{-\lambda t} \Phi(t,0)$ and $\Phi$ given by \eqref{wron}. 
	For $i:=2$, the solution 
	$$v_2(t)=R(t,0)v_2(0)-\frac{t}{T\mu}v_{1}(t)$$ is directly obtained from the formula: $$v_2(t)=R(t,0)v_2(0)-\frac{1}{T\mu}\int_0^tR(t,s)v_{1}(s)ds.$$
	As $R(T,0)=\mu^{-1} \Phi(T,0)$, it follows that:
	$$v_{1}(T)+\mu v_2(T)= \Phi(T,0)v_2(0)$$ where $v_{1}(T)=v_{1}(0)$.
	Thus if $v_2(0):=\phi_2$ then it follows that $v_2(T)=v_2(0)$ which proves that $v_2(t)$ is $T-$periodic.\\
	Now, assume that this property holds recursively until the index $i-1$ and:
	$$v_{i-1}(t)=R(t,0)v_{i-1}(0)-\frac{t}{T\mu}v_{i-2}(t)$$ 
	then, as $$
	v_i(t)=R(t,0)v_i(0)-\frac{1}{T\mu}\int_0^tR(t,s)v_{i-1}(s)ds,$$
	it is straightforward to show that:
	$$v_i(t)=R(t,0)v_i(0)-\frac{t}{T\mu}v_{i-1}(t).$$
	Thus, following the same reasoning as before, the conclusion on the periodicity of $v_i$ follows by setting $v_i(0):=\phi_i$. 
\end{proof}

We are now in position to give a closed form formula for a Floquet Factorization.

\begin{theorem}\label{diago} Assume that the $T-$periodic function $A(t)$ belongs to $L^2([0 \ T])$ and let $ \Phi(T,0)$ is given by \eqref{wron}. Consider for $i:=1,\cdots,n$, the eigenvalues $\mu_i$ and the generalized eigenvectors $\phi_i$ of $ \Phi(T,0)$ and set $\lambda_i$ to the principal value of $\frac{1}{T}\log(\mu_i)$. Consider for each $\lambda_i$, the solution $v_i$ of the initial value problem with $v_i(0):=\phi_i$, provided by Theorem~\ref{nondef} (or \ref{def} if $ \Phi(T,0)$ is defective). 

Then, {a} Floquet factorization is determined by $W(t):=[v_1(t),\cdots, v_n(t)]$ and $Q:=\Lambda$ where $\Lambda$ is a Jordan normal form given by, for $i:=1,\cdots,n$, $\Lambda(i,i):=\lambda_i$, $\Lambda(i,i+1):=0$ or $\frac{1}{T\mu_i}$ and zeros elsewhere. Moreover, the $T-$periodic and absolutely continuous matrices $W$ and $W^{-1}$ satisfy:
 \begin{align}
	\dot W(t)&=A(t)W(t)-W(t)\Lambda \ a.e. \label{dP}\\
	\dot W^{-1}(t)&=-W^{-1}(t)A(t)+\Lambda W^{-1}(t) \ a.e. \label{dP2}\end{align}
and the operator $\mathcal{W}:=\mathcal{T}(W)$ is bounded on $\ell^2$, invertible and satisfies the eigenvalue problem:
	\begin{equation}\mathcal{W}^{-1}(\mathcal{A}-\mathcal{N})\mathcal{W}= \Lambda \otimes \mathcal{I}- \mathcal{N} \label{eigenprob}\end{equation}
In addition, taking $z(t):=W^{-1}(t)x(t)$ transforms the LTP system $\dot x=A(t)x$ into the LTI system
	\begin{equation}
		\dot z=\Lambda z\ a.e.\label{lti}\end{equation} 
\end{theorem}
\begin{proof} For $i:=1,\cdots,n$, consider for $\lambda_i$ the principal value of $\frac{1}{T}\log(\mu_i)$ only and let the $T-$periodic vectors $v_i$ determined using Theorem~\ref{nondef} (or Theorem~\ref{def} if $ \Phi(T,0)$ is defective) with $V_i:=\mathcal{F}(v_i)$. Denote by $$A_{k}:=\left(\begin{array}{cccc}A_{11,k} & A_{12,k} & \cdots & A_{1n,k} \\A_{21,k} & A_{22,k} & & \\\vdots & & \ddots & \vdots \\A_{n1,k} & \cdots & \cdots & A_{nn,k}\end{array}\right)$$ the $k-th$ phasors of $A$ and $V_{i,k}$ the $k-th$ phasors of $V_i$. We have for any $k$:
	$$\sum_{p\in\mathbb{Z}} A_{k-p}V_{i+1,p}-j\omega k V_{i+1,k}=\lambda_i V_{i+1,k}+s_iV_{i,k}$$
	with $s_i:=\frac{1}{T\mu_i}$ or $0$.
	Thus, a $m$-shift in the components of $V_i$, {$V_{i+1}$ }leads to: 
	\begin{align*}
\sum_{p\in\mathbb{Z}} A_{k-p}V_{i+1,p+m}&-j\omega k V_{i+1,k+m}\\
	&=(\lambda+j\omega m) V_{i+1,k+m}+{s_i} V_{i,k+m}	\end{align*}
	{which means that the $m$-shifted vector is also a generalized eigenvector associated to $\lambda_i+j\omega m$ (and not to $\lambda_i$). 
It follows that $\mathcal{T}(v_i)$ is the set of all generalized eigenvectors associated to all values of $\frac{1}{T}\log(\mu_i)$ defined modulo $j\omega$.}

Now, set $W:=[v_1,\cdots,v_n]$. Obviously, $W(t)$ satisfies \eqref{dP}.
	As $W(t)$ is a $T-$periodic and absolutely continuous matrix function, $\mathcal{W}:=\mathcal{T}(W)=[\mathcal{T}(v_1),\cdots,\mathcal{T}(v_n)]$ is a constant and bounded operator on $\ell^2$ (see Theorem~\ref{borne}). Furthermore,  {using similar steps as in the proof of (\cite{Blin}, Theorem 5)}, the $n\times n$ block Toeplitz matrix $\mathcal{W}$ satisfies :
	$$(\mathcal{A}-\mathcal{N})\mathcal{W}= \mathcal{W}(\Lambda \otimes \mathcal{I}- \mathcal{N}).$$ 
	As $\mathcal{W}$ solves the eigenvalue problem {for all admissible eigenvalues} and is invertible, the same holds true for $W(t)$.
	Since $\mathcal{W}^{-1}(\mathcal{A}-\mathcal{N})= (\Lambda \otimes \mathcal{I}- \mathcal{N})\mathcal{W}^{-1}$,   {using similar steps as in the proof of (\cite{Blin}, Theorem 5)}, it is straightforward to establish that the absolutely continuous matrix function $W^{-1}:= \mathcal{T}^{-1}(\mathcal{W}^{-1})$ satisfies \eqref{dP2}.

Finally, let $x$ be a solution of $\dot x =A(t)x$ in Carath\'eodory sense and set $z(t):=W^{-1}(t)x(t)$. From \eqref{dP2} we have:
	\begin{align*}
		\dot z(t)&=\dot W^{-1}(t)x(t)+W^{-1}(t)\dot x(t) \ a.e.\\
		&=\Lambda z(t)\ a.e.
	\end{align*}
\end{proof}
 
\begin{corollary}$(\mathcal{A}-\mathcal{N})$ is non-defective if and only if $\Phi(T,0)$ is non-defective
\end{corollary}
\begin{proof}As $z(t)=e^{\Lambda t}z_0$ and as $z(t):=W^{-1}(t)x(t)=W^{-1}(t)\Phi(t,0)W(0)z_0$, it follows that $$e^{\Lambda t}=W^{-1}(t)\Phi(t,0)W(0).$$ 
For $t:=T$ since $W$ is $T-$periodic, we have $W(0)=W(T)$ and \begin{equation}e^{\Lambda T}=W^{-1}(T)\Phi(T,0)W(T).\label{change}\end{equation}
Now if $(\mathcal{A}-\mathcal{N})$ is non-defective, the eigenvalue problem corresponding to \eqref{eigenprob} is determined by a diagonal matrix $\Lambda$. Thus,  $e^{\Lambda T}$ is diagonal and we conclude from \eqref{change} that $\Phi(T,0)$ is non defective.
Reciprocally, if $\Phi(T,0)$ is non-defective, Theorem~\ref{diago} leads to \eqref{eigenprob} with $\Lambda$ diagonal.
\end{proof}
\color{black}

The previous Theorem provides a simple characterization of a Floquet factorization which is of interest for analysis purpose. The fact that the input harmonic matrix remains a full matrix when applying {a} Floquet factorization makes this approach difficult to apply to design stabilizing state feedback control laws for example. Our choice is to push further the spectral analysis of the operator $(\mathcal{A}-\mathcal{N})$ and analyze the impact of a $m-$truncation on the spectrum of $(\mathcal{A}_m-\mathcal{N}_m)$ in order to provide efficient algorithms that can also be used for harmonic control design. We start by the following corollary which states that the spectrum of $(\mathcal{A}-\mathcal{N})$ is unbounded and discrete.
\begin{corollary}\label{corlll} Assume that $(\mathcal{A}-\mathcal{N})$ is non-defective. The spectrum of $(\mathcal{A}-\mathcal{N})$ is given by the unbounded and discrete set
	$$\sigma(\mathcal{A}-\mathcal{N}):=\{ \lambda_p +j\omega k: k\in\mathbb{Z}, p:=1,\cdots, n \}$$ where $\lambda_p$, $p:=1,\cdots, n$ are not necessarily distinct eigenvalues. 
\end{corollary}
\begin{proof}
	Consider the unbounded diagonal operator $D:=(\Lambda \otimes \mathcal{I}- \mathcal{N})$.
	As the point spectrum $\sigma_{ps}:=\{ \lambda_p +j\omega k: k\in\mathbb{Z}, p=1,\cdots, n \}$ has no cluster points, it is a closed set and $\sigma_{ps}\subset\sigma(D)$. Denote by $D_j$ the $j-$th entry of the diagonal of $D$ and $e_j$ the $j-th$ vector of the basis. If $\zeta \notin \sigma_{ps}$, then $S$ defined by
	$Se_j:=\frac{1}{\zeta-D_j}e_j $ for any $j\in\mathbb{Z}$ is a bounded (diagonal) operator on $\ell^2$ whose inverse is $\zeta \mathcal{I}-D$. Thus, $\zeta \notin\sigma(D)$ and $\sigma(D)=\sigma_{sp}$.
\end{proof}

The result of Corollary~\ref{corlll} holds also when $(\mathcal{A}-\mathcal{N})$ is defective. This can be established by defining $\mathcal{J}:=(J\otimes \mathcal{I}-\mathcal{N})$ with $J$ a Jordan normal form and showing that $\zeta \mathcal{I}-\mathcal{J}$ is invertible for any $\zeta \notin \sigma_{ps}$. The invertibility of $\zeta \mathcal{I}-\mathcal{J}$ is proved recursively on the $n\times n$ blocks of $\zeta \mathcal{I}-\mathcal{J}$ noticing that each of these blocks is diagonal and using 
	the matrix formula $$\left(\begin{array}{cc}A & B \\0 & C\end{array}\right)^{-1}=\left(\begin{array}{cc}A^{-1} &-A^{-1} BC^{-1} \\0 & C^{-1}\end{array}\right).$$

We discuss the properties of the inverse of $(\mathcal{A}-\mathcal{N})$ in the following corollary.

\begin{corollary}\label{spectre_bounded} $\Lambda$ is invertible if and only if the operator $(\mathcal{A}-\mathcal{N})$ is invertible. Moreover, $(\mathcal{A}-\mathcal{N})^{-1}$ is bounded on $\ell^2$ and $$\sigma^+=\|(\mathcal{A}-\mathcal{N})^{-1}\|_{\ell^2}$$ where
	$\sigma^+:=\sup\{|(\lambda_p +j\omega k)^{-1}|: k\in\mathbb{Z}, p=1,\cdots, n\}.$
\end{corollary}
\color{black}
\begin{proof}The result follows from the above theorem and noticing that the $\ell^2$ operator norm corresponds to the maximun singular value.
\end{proof}

\begin{remark}
	As shown in \cite{Blin}, if $x$ is the solution de $\dot x=A(t)x$ then $X:=\mathcal{F}(x)\in C^0(\mathbb{R},\ell^2)$ and we have $(\mathcal{A}-\mathcal{N})X(t) \notin \ell^{p}$, for any $1<p<+\infty$ and $(\mathcal{A}-\mathcal{N})X(t) \in \ell^{\infty}$. Clearly, $(\mathcal{A}-\mathcal{N})$ is not a bounded operator on $l^2$ while its inverse (if it exists) is bounded.
\end{remark}
 \subsection{Spectrum analysis of $(\mathcal{A}_m-\mathcal{N}_m)$}
Here, we explain how the spectrum of $(\mathcal{A}_m-\mathcal{N}_m)$ is modified w.r.t. the spectrum of $(\mathcal{A}-\mathcal{N})$ when performing a $m-$truncation on $(\mathcal{A}-\mathcal{N})$. From now, we assume that the $T-$periodic matrix function $A(t)$ belongs to $L^\infty(\mathbb{R},\mathbb{R}^{n\times n})$ or equivalently $\mathcal{A}$ is a bounded operator on $\ell^2$. This will help us in providing algorithms with guarantees at an arbitrarily small error when a $m-$truncation is applied. For simplicity reasons, we provide the results when the operator $(\mathcal{A}-\mathcal{N})$ is non-defective but the results hold true in general.
\begin{theorem}\label{spectre_plus} Assume that $A(t)\in L^\infty([0 \ T])$ and $(\mathcal{A}-\mathcal{N})$ is non-defective. Denote by $\sigma:=\{\lambda_p+j\omega k: k\in\mathbb{Z}, p:=1,\cdots,n\}$ the spectrum of $(\mathcal{A}-\mathcal{N})$. Let $(\mathcal{A}_{m^+}-\mathcal{N}_{m^+})$ be a left $m-$truncation of $(\mathcal{A}-\mathcal{N})$ according to Definition~\ref{truncdef} and assume that it is non defective, with an eigenvalues set denoted by $\Lambda_{m^+}$. 
\begin{enumerate}
\item For $\epsilon>0$, there exists an index $j_0$ such that for any eigenvalue $\lambda \in \Lambda_1^+(m):=\{\lambda_p+j\omega k: k\in\mathbb{Z}, k\leq m+1 -j_0, p:=1,\cdots,n\}\subset \sigma$:
	\begin{align}\|(\mathcal{A}_{m^+}-\mathcal{N}_{m^+}-\lambda \mathcal{I}_{m^+}) {V}|_{m^+}\|_{\ell^{2}}<\epsilon \label{e3}\end{align}
	where ${V}|_{m^+}$ is the left $m-$truncation of the eigenvector associated to $\lambda$. 
\item The set $\Lambda_{m^+}$ can be approximated by the union of $\Lambda_1^+(m)$ and $\Lambda_2^+(m)$
	that is $$\Lambda_{m^+}\approx\Lambda_1^+(m)\cup \Lambda_2^+(m)$$
	where $\Lambda_2^+(m)$ is a finite subset of $\Lambda_{m^+}$.
	Moreover, any eigenvalue $\lambda_{m+1}$ which belongs to the set $\Lambda_2^+(m+1)$ is obtained by the relation: 
	$\lambda_{m+1}=\lambda_{m}+j\omega$ 
	where $\lambda_{m}$ belongs to $\Lambda_2^+(m)$. 
	
\end{enumerate}
\end{theorem}
\begin{proof} It is sufficient to prove the theorem for $n:=1$. Indeed, as the difficulties are related to infinite Toeplitz matrices, if the result is established for $n:=1$, the same result holds for any finite $n$ using ad hoc formula and Proposition~\ref{general}. 
	Consider $\mathcal{W}$ given by \eqref{eigenprob}. When $n:=1$, the set of eigenvalues is given by $\{ \lambda +j\omega k: k\in\mathbb{Z}, \}$ for a given $\lambda$ and the matrix $\mathcal{W}$ reduces to $\mathcal{W}:=\mathcal{V}=\mathcal{T}(v)$ with phasors denoted by $V$.
	Applying a left $m-$truncation and using Theorem~\ref{product}, we have: 
	$$(\mathcal{A}_{m^+}-\mathcal{N}_{m^+})\mathcal{V}_{m^+}-\mathcal{V}_{m^+}(\lambda \mathcal{I}-\mathcal{N}_{m^+})=E^+$$
	where $E^+:=-\mathcal{H}(A^+)\mathcal{H}(V^-)$.
	
	For $j:=1,2,\cdots$, the $j-th$ column of $E^+$ is provided by
	$E^+(j):=-\mathcal{H}(A^+)V|_{j}^-$
	where $V|_j^-:=(V_{-j},V_{-j-1},\cdots, )$.
	Using Theorem~\ref{borne}, we have:
	\begin{align*}\|E^+(j)\|_{\ell^2}&\leq\|\mathcal{H}(A^+)\|_{\ell^2}\|V|_{j}^-\|_{\ell^2}\\
		&\leq\|A(t)\|_{L^\infty}\|V|_{j}^-\|_{\ell^2}.\end{align*}
	Therefore, for a given $\epsilon>0$, there always exists an index $j_0$ such that $\|E^+(j)\|_{\ell^2}\leq \epsilon $ for $j\geq j_0$ since $\|V|_{j}^-\|_{\ell^2}\rightarrow 0$ when $j \rightarrow +\infty$
	which establishes \eqref{e3}. The remaining $j_0-1$ eigenvalues form a finite subset $\Lambda_2^+(m)$ of $\Lambda_{m^+}$. Thus, if $\lambda_2\in\Lambda_2^+(m)$ is an eigenvalue associated to its semi-infinite eigenvector $V_{\lambda_2}$, it follows that: 
	$$(\mathcal{A}_{m^+}-\mathcal{N}_{m^+})V_{\lambda_2}=\lambda_2V_{\lambda_2}$$
	and it is straightforward to show that
	$$(\mathcal{A}_{(m+1)^+}-\mathcal{N}_{(m+1)^+})V_{\lambda_2}=(\lambda_2+j\omega) V_{\lambda_2}.$$
	Therefore, any eigenvalue of the set $\Lambda_2^+(m+1)$ is obtained from $\lambda_2 \in\Lambda_2^+(m)$ by adding $j\omega$ and the associated semi-infinite eigenvector is obtained by shifting $V_{\lambda_2}$. 
\end{proof}
\begin{theorem}\label{spec_trunc}Assume that the matrix $A(t)\in L^\infty([0 \ T])$ and that $(\mathcal{A}-\mathcal{N})$ is non-defective with $\sigma:=\{\lambda_p+j\omega k: k\in\mathbb{Z}, p:=1,\cdots,n\}$ its spectrum. 
	Assume that the ${m}-$truncation $(\mathcal{A}_{{m}}-\mathcal{N}_{{m}})$ is non-defective with its eigenvalues set denoted by $\Lambda_{ m}$. 
	
	For $\epsilon>0$, there exists a $m_0$ such that for $m\geq m_0$:
	\begin{enumerate}
		\item there exists an index $j_0$ such that for any eigenvalue $\lambda_1\in \Lambda_1(m)$ defined by the subset $\{\lambda_p+j\omega k: |k|\leq j_0, p:=1,\cdots,n\}$ of $\sigma$, the following relation is satisfied:
		\begin{align}\|(\mathcal{A}_{{m}}-\mathcal{N}_{{m}}-\lambda_1 \mathcal{I}_{{m}}){V_1}|_{m}\|_{\ell^{2}}<\epsilon \label{e4}\end{align}
	 where ${V_1}|_{m}$ is the $m-$truncation of the eigenvector associated to $\lambda_1$.
		\item for any eigenvalue $\lambda_2\in \Lambda_2(m^+)$ or in $ \Lambda_2(m^-):= \bar{\Lambda}_2(m^+)$ with $\Lambda_2(m^+)$ defined in Theorem~\ref{spectre_plus},
		the following relation is satisfied:
		\begin{align*}\|(\mathcal{A}_{{m}}-\mathcal{N}_{{m}}-\lambda_2 \mathcal{I}_{{m}}) {V}_2|_m\|_{\ell^{2}}<\epsilon \end{align*}
		where ${V_2}|_{m}$ is the $m-$truncation of the eigenvector associated to $\lambda_2$.
	\end{enumerate}
	Then, the set $\Lambda_{m}$ can be approximated by the union of the sets $\Lambda_1( m)$, $\Lambda_2^-( m)$ and $\Lambda_2^+( m)$ that is:
	$$\Lambda_{m}\approx\Lambda_1( m)\cup \Lambda_2^-( m)\cup \Lambda_2^+( m).$$
\end{theorem}
\begin{proof} As before, the proof is given for $n:=1$. The right truncation leads to a symmetric result of Theorem~\ref{spectre_plus} for which $\Lambda_2^-(m)=\bar \Lambda_2^+(m)$. In case both right and left $m-$truncations are perfomed, applying Theorem~\ref{product} to \eqref{eigenprob} leads to : 
	$$(\mathcal{A}_{{m}}-\mathcal{N}_{{m}})\mathcal{V}_{{m}}-\mathcal{V}_{{m}}(\lambda \mathcal{I}_{{m}}-\mathcal{N}_{{m}})=-E_{m}$$
	where $E_{m}:=E^+_{m}+E^-_{m}$ and
	\begin{align*}
		E^+_{m}&:=\mathcal{H}_{({m},\eta)}(A^+)\mathcal{H}_{(\eta,{m})}(V^-) \\
		E^-_{m}&:=J_{m}\mathcal{H}_{({m},\eta)}(A^-)\mathcal{H}_{(\eta,{m})}(V^+) J_{m}. 
	\end{align*}
	Notice that $E^-_m(i,j)$ is simply obtained from $E^+_m(i,j)$ by a central symmetry of index $(m+1,m+1)$.
	As in the previous proof, for $j:=1,2,\cdots$, the $\ell^2$ norm of the $j-th$ column of $E^+$ satisfies
	$$\|E^+(j)\|_{\ell^2}\leq\|A(t)\|_{L^\infty}\|V|_{j}^-\|_{\ell^2}$$
	and for a given $\epsilon >0$, there exists an index $j_0<m$ (provided that $m$ is chosen sufficiently large) such that for any $j\geq j_0$
	$$\|E^+(j)\|_{\ell^2}\leq \epsilon/2.$$
	By symmetry, the $\ell^{2}$-norm of the columns of $E^-_m$ is less than $\epsilon/2$ for $j\leq 2(m+1)-j_0$.
	Thus, it follows that the $\ell^{2}$-norm of the columns of $E^-_m+E^-_m$ is less than $\epsilon$ for $j$ such that $ j_0\leq j\leq 2(m+1)-j_0$.
	Consequently, Equation \eqref{e4} is satisfied.
	
	Now it remains to show that the elements of $\Lambda_2^+(m)$ and $\Lambda_2^-(m)$ are eigenvalues of $(\mathcal{A}_{m}-\mathcal{N}_{m})$ up to an arbitrary small error.
	If $\lambda_2\in\Lambda_2^+(m)$ is an eigenvalue associated to an eigenvector $V_2$, it follows that: 
	\begin{equation}(\mathcal{A}_{m^+}-\mathcal{N}_{m^+})V_2=\lambda_2V_2. \label{e5}\end{equation}
	If a right $m-$truncation is applied on \eqref{e5}, then 
	$$(\mathcal{A}_{m}-\mathcal{N}_{m})V_2|_m-\lambda_2V_2|_m= E_2$$
	where 
	$E_2:=- J_m \mathcal{H}_{(m,\infty)}(A^-)V_2|_m^+$ (see \eqref{z2} in Proposition~\ref{product3})
	and $V_2|_m^+:=(V_{2_{m+1}},V_{2_{m+2}},\cdots)$. As before, we have:
	$\|E_2\|_{\ell^2}\leq \|A(t)\|_{L^\infty}\|V_2|_m^+\|_{\ell^2}$.
	Hence, there always exists a $m_0$ such that for $m\geq m_0$, 
	$$\|E_2\|_{\ell^2}\leq \epsilon$$
	since $\|V_2|_m^+\|_{\ell^2}\rightarrow 0$ when $m\rightarrow +\infty$.
	This completes the proof.
\end{proof}

\begin{corollary}\label{specbound}Assume that the matrix $A(t)\in L^\infty([0 \ T])$.
	If $(\mathcal{A}-\mathcal{N})$ is invertible, there exists a $m_0>0$ such that for any $m\geq m_0$, the matrix $(\mathcal{A}_{m}-\mathcal{N}_{m})$ is invertible. Moreover, $\|(\mathcal{A}_{m}-\mathcal{N}_{m})^{-1}\|_{\ell^2}$ is uniformly bounded i.e. $\sup_{m\geq m_0} \|(\mathcal{A}_{m}-\mathcal{N}_{m})^{-1}\|_{\ell^2}<+\infty.$
\end{corollary}
\begin{proof}
	As $\lim_{m\rightarrow+\infty} |\min(\Lambda_2^+(m))|=+\infty$, and as $(\mathcal{A}-\mathcal{N})$ is invertible, for sufficiently large $m$, $(\mathcal{A}_m-\mathcal{N}_m)$ is not singular and the eigenvalues of $(\mathcal{A}_m-\mathcal{N}_m)^{-1}$ are uniformly bounded by $\sup\{|(\lambda_p -j\omega k)^{-1}|: k\in\mathbb{Z}, p:=1,\cdots, n\}.$ Thus, $\|(\mathcal{A}_{m}-\mathcal{N}_{m})^{-1}\|_{\ell^2}$ is uniformly bounded.
\end{proof}
\subsection{Example}
Consider the following $2\times 2$ block Toeplitz matrix 
$$\mathcal{A}:=\left(\begin{array}{cccc}
	\mathcal{A}_{11} & \mathcal{A}_{12} \\
	\mathcal{A}_{21} & \mathcal{A}_{22} \\
\end{array}\right)$$ 
where the Toeplitz matrices $\mathcal{A}_{ij}$ are characterized by
\begin{align*}
	a_{11}&:=(0.5, 0.6-j,\underline{ -1},0.6+j ,0.5)\\
	a_{12}&:=(1.3-0.4j,-2.2+0.5j,\underline{-0.4},-2.2-0.5j,1.3+0.4j)\\
	a_{21}&:=(-0.3 - 0.6j, 0.4+0.7j, \underline{-0.1}, 0.4-0.7j, -0.3 + 0.6j)\\
	a_{22}&:=(-1.3-1.8j, 1.4-1.6j, \underline{-2},1.4+1.6j,-1.3+1.8j)
\end{align*} 
with the underlined terms corresponding to the index $0$. 
\begin{figure}[h]
	\begin{center}
		\includegraphics[width=\linewidth,height=5cm]{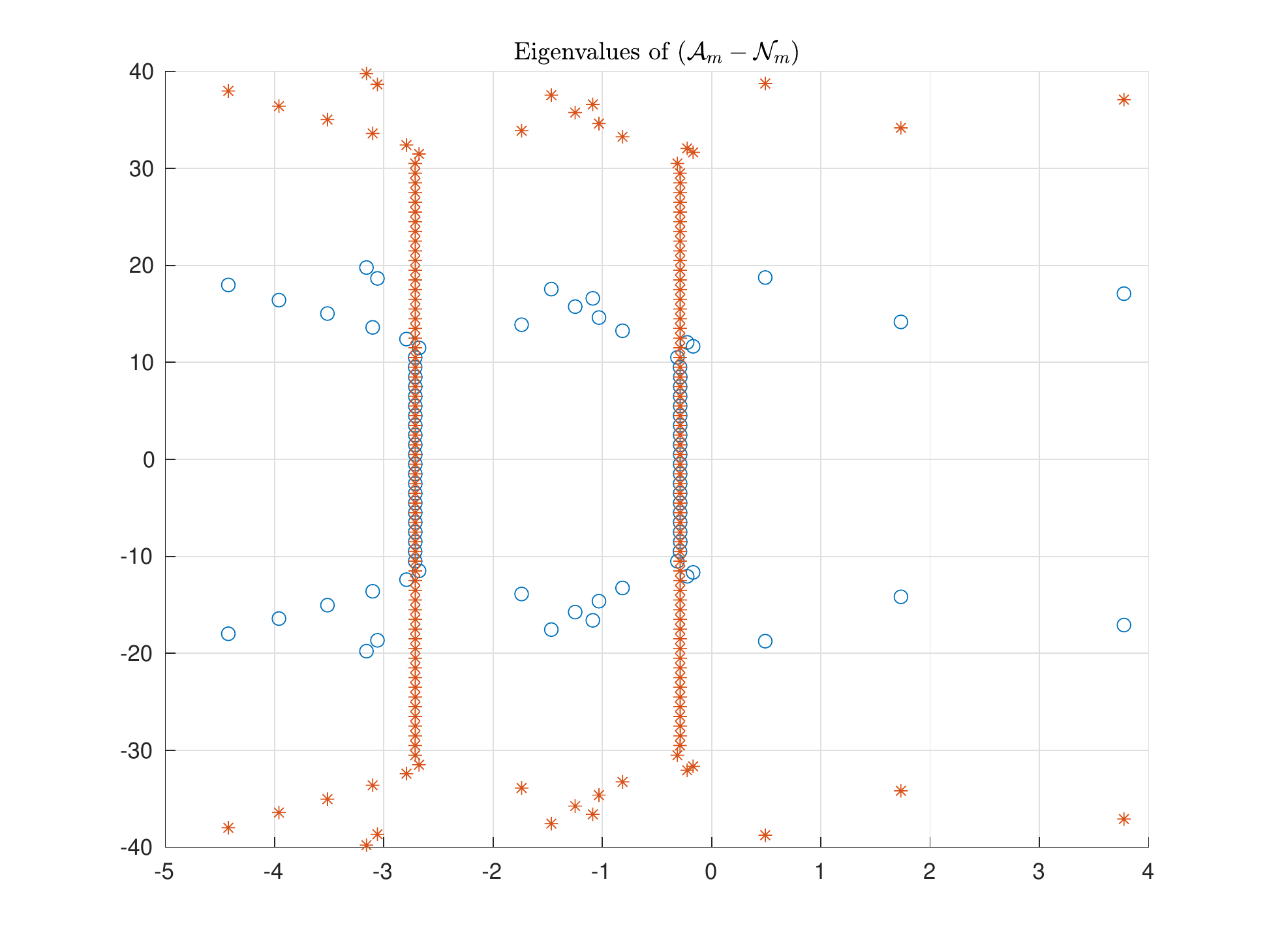}
		\caption{Eigenvalues of $(\mathcal{A}_m-\mathcal{N}_m)$ for $m:=20$ (blue) or $40$ (red).}\label{fig3}
	\end{center}
\end{figure}
The eigenvalues of $(\mathcal{A}_m-\mathcal{N}_m)$ are depicted in Fig.~\ref{fig3} for $m:=20$ (green circles) and $m:=40$ (red stars). We clearly observe the sets $\Lambda_1(m)$, $\Lambda_2^-(m)$ and $\Lambda_2^+(m)$. Notice that $\Lambda_1(m)$ is defined by two eigenvalues (as expected) satisfying $R(\lambda_1)\approx-0.3$ and $R(\lambda_2)\approx-2.7$ as shown by the alignment of the eigenvalues along these vertical axes. Thus, $(\mathcal{A}-\mathcal{N})$ is Hurwitz while $(\mathcal{A}_m-\mathcal{N}_m)$ is never Hurwitz for all $m$ since $\Lambda_2^-(m)$ and $\Lambda_2^+(m)$ have eigenvalues with positive real parts.
As mentioned, we see that the set $\Lambda_2^+(40)$ is obtained from the set $\Lambda_2^+(20)$ by a translation of $j\omega 20$ ($\omega:=1$).

This example illustrates the fact that having $(\mathcal{A}-\mathcal{N})$ Hurwitz, if we try to solve a Lyapunov equation with a truncated version $(\mathcal{A}_m-\mathcal{N}_m)$, the solution would never be positive definite whatever $m$. This motivates the following section devoted to solving harmonic Lyapunov equations.

\section{Solving Harmonic Lyapunov equation}
Taking benefit from the spectral analysis of the previous section, 
the objective here is to study how the infinite dimensional harmonic Lyapunov equation \eqref{al} can be solved in a practice without invoking {a} Floquet factorization. 
The following proposition introduces the symbol Lyapunov equation. 
\begin{proposition}\label{symLyap} Assume that $A(t)\in L^\infty([0 \ T])$. The symbol harmonic Lyapunov equation is given by 
	\begin{equation}A(z)' P(z)+P(z)A(z)+(\mathbbm{1}_{n,n}\otimes N(z))\cdot P(z)+Q(z)=0 \label{sym_syl}\end{equation}
	where the symbol matrices $A(z)$, $Q(z)$, $P(z)$ are given by \eqref{symA} and where $N(z):=\sum_{k=-\infty}^{+\infty}j\omega k z^k$. 
\end{proposition}
\begin{proof}
	Consider the harmonic Lyapunov equation \eqref{al}. It is straightforward to show that the product $ -\mathcal{N}^*\mathcal{P}-\mathcal{P}\mathcal{N}=\mathcal{N}\mathcal{P}-\mathcal{P}\mathcal{N}$ is formed by $ n\times n$ blocks of Toeplitz matrices whose symbols for $i,j:=1,\cdots, n $, are given by the Hadamard product $N(z)\cdot P_{ij}(z)$ where $P_{ij}(z)$ refers to the symbol associated the entry $(i,j)$ of the matrix $P(t)$. Consequently, the symbol associated to $ \mathcal{N}\mathcal{P}-\mathcal{P}\mathcal{N}$ is $$(\mathbbm{1}_{n,n}\otimes N(z))\cdot P(z)$$ where $N(z):=\sum_{k=-\infty}^{+\infty}j\omega k z^k$. Replacing the Toeplitz matrix by its symbol and noticing that the symbols associated to $\mathcal{A}^*$ are given by the transpose of $A(z)$ ends the proof.
\end{proof}
Looking at the previous symbol Lyapunov equation, we see that it is not possible to factorize $P(z)$ to obtain a solution. In the next theorem, we show that if we try to solve a truncated version of \eqref{al}, the resulting solution is not Toeplitz. As this Toeplitz property is required for the infinite dimension case, an important practical consequence is the fact that the time counterpart $P(t)$ does not exist and cannot be reconstructed using (\ref{recos}). For a better understanding, we show in the following theorem that the solution $P_m$ obtained by solving the truncated harmonic Lyapunov equation differs from the solution of the infinite-dimensional harmonic Lyapunov equation by a correcting term $\Delta P_m$. 
\begin{theorem}\label{finite_hle_m1}
	Consider finite dimension Toeplitz matrices $\mathcal{A}_m:=\mathcal{T}_m(A)$ and $\mathcal{Q}_m:=\mathcal{T}_m(Q)$. The solution $P_m$ of the Lyapunov equation
	\begin{equation}(\mathcal{A}_m-\mathcal{N}_m)^*P_m+P_m(\mathcal{A}_m-\mathcal{N}_m)+\mathcal{Q}_m=0\label{trunc_lyap_m1}\end{equation} is given by 
	$P_m:=\mathcal{P}_m+\Delta P_m$ where $\mathcal{P}_m:=\mathcal{T}_m(P)$ with $P(z)$ solution of \eqref{sym_syl} and $\Delta P_m$ satisfies:
	$$(\mathcal{A}_m-\mathcal{N}_m)^*\Delta P_m+\Delta P_m(\mathcal{A}_m-\mathcal{N}_m)=E^++E^-$$
with
	\begin{align*}
		E^+:=\mathcal{H}_{(m,\eta)}(A^+)&\mathcal{H}_{(\eta,m)}(P^-) +\mathcal{H}_{(m,\eta)}(P^+)\mathcal{H}_{(\eta,m)}(A^-)\\
		E^-:=\mathcal{J}_{n,m}(\mathcal{H}_{(m,\eta)}&(A^-)\mathcal{H}_m(P^+) \\&+\mathcal{H}_{(m,\eta)}(P^-)\mathcal{H}_{(\eta,m)}(A^+))\mathcal{J}_{n,m} 
	\end{align*}
and $2\eta \geq \min d^o(A(z),P(z))$.
\end{theorem}
\begin{proof}
	The proof is obvious from Theorem~\ref{product} and noticing that  $-\mathcal{N}_m^*{P}_m-{P}_m\mathcal{N}_m$ does not give rise to a correction term since $\mathcal{N}$ is a diagonal matrix.
\end{proof}
In practice, it is not clear how the Toeplitz part $\mathcal{P}_m$ of $P_m$ can be extracted since the symbol $P(z)$ is implicitly given by \eqref{sym_syl}. In fact, it can be shown that this linear problem is rank deficient and has infinitely many solutions. Thus, our aim is to prove that $\mathcal{P}$ can be determined up to an arbitrary small error.
The necessity to determine $\mathcal{P}$ up to an arbitrary small error instead of $P_m$ is crucial to prove stability of
$(\mathcal{A}-\mathcal{N})$. This is due to the fact that the matrix $(\mathcal{A}_m-\mathcal{N}_m)$ would never be Hurwitz for any $m$ when $(\mathcal{A}-\mathcal{N})$ is Hurwitz.

\begin{theorem}\label{sylvester_approx}Assume that $(\mathcal{A}-\mathcal{ N})$ is invertible.
	 {The phasor ${\bf P}:=\mathcal{F}(P)$ associated to the solution $\mathcal{P}:=\mathcal{T}(P)$} of the infinite-dimensional harmonic Lyapunov equation (\ref{al}) is given by:
	\begin{equation}col({\bf P}):=-(Id_n \otimes (\mathcal{A}-\mathcal{ N})^*+Id_n\circ\mathcal{A}^*)^{-1}col({\bf Q})\label{inf_sym}\end{equation}
	where {\small\begin{equation}Id_n\circ\mathcal{A}:=\left(\begin{array}{cccc} 
				Id_n\otimes \mathcal{A}_{11} & Id_n\otimes\mathcal{A}_{12} & \cdots & Id_n\otimes\mathcal{A}_{1n} \\
				Id_n\otimes\mathcal{A}_{21} & Id_n\otimes\mathcal{A}_{22} & & \vdots \\
				\vdots & & \ddots& \vdots \\
				Id_n\otimes\mathcal{A}_{n1} & \cdots & \cdots & Id_n\otimes\mathcal{A}_{nn}\end{array}\right)\label{symB}\end{equation}}
	with $\mathcal{N}$ given by \eqref{N} and where the matrix ${\bf Q}:=\mathcal{F}(Q)$.
	
\end{theorem}
\begin{proof}
	Applying the well known formula $(Id_n \otimes A +B'\otimes Id_m)col(X)=col(C)$ associated to the Sylvester equation $AX+XB=C$
	to the case of the symbol Lyapunov equation~\eqref{sym_syl}, one gets:
	\begin{align*}(Id_n \otimes A(z)'& +A(z)'\otimes Id_n)col(P(z))\\
		&+(\mathbbm{1}_{n^2,1}\otimes N(z))\cdot col(P(z))&=-col(Q(z))\end{align*}
	Notice that $col((\mathbbm{1}_{n,n}\otimes N(z))\cdot P(z))=col(\mathbbm{1}_{n,n}\otimes N(z))\cdot col(P(z))
	=(\mathbbm{1}_{n^2,1}\otimes N(z))\cdot col(P(z)) $. Observe that the $i-$th lines, $i=1,\cdots,n^2$, of this multi-polynomial equation is given by:
	\begin{equation*}\sum_{j=1}^{n^2}M_{ij}(z)P_j(z)+N(z)\cdot P_i(z)=-Q_i(z)\end{equation*}
	where $P_j(z)$, $j:=1,\cdots ,n^2$ refers to the components of $col(P(z))$ and where the terms $M_{ij}(z):=(Id_n \otimes A(z)' +A(z)'\otimes Id_n)_{ij}$ are determined from the expansions:
{\small$$Id_n \otimes A(z)'=\left.\left(\begin{array}{cccc}A'(z) & 0 & 0 & 0 \\0 & A'(z) & 0 & 0 \\0 & 0 & \ddots & 0 \\0 & 0 & 0 & A'(z)\end{array}\right)\right\} n \text{ times}$$} with $A(z)$ provided by \eqref{symA} and
	{\small$$A(z)'\otimes Id_n=\left.\left(\begin{array}{cccc}
			A_{11}(z)Id_n & A_{21}(z)Id_n & \cdots & A_{n1}(z)Id_n \\
			A_{12}(z)Id_n & A_{22}(z)Id_n & & \\
			\vdots & & \ddots & \vdots \\
			A_{1n}(z)Id_n & \cdots & \cdots & A_{nn}(z)Id_n\end{array}\right)\right..$$}
	
	Recall that $\mathcal{N}:=Id_n\otimes diag(j\omega k,\ k\in \mathbb{Z})$. The symbol $(\mathbbm{1}_{n^2,1}\otimes N(z))$ has coefficients corresponding to the matrix $Id_n \otimes \mathcal{ N}=-Id_n \otimes \mathcal{ N}^*$.  {Replacing each symbol $A_{ij}(z)$ in the above equation with their associated Toeplitz matrix leads to an equivalent equation involving the coefficients:}
	\begin{equation*}(Id_n \otimes (\mathcal{A}-\mathcal{ N})^*+Id_n\circ\mathcal{A}^*)col({\bf P})=-col({\bf Q})\end{equation*}
	where $Id_n\circ\mathcal{A}$ is given by \eqref{symB}, ${\bf P}:=\mathcal{F}(P)$ and ${\bf Q}:=\mathcal{F}(Q)$.
	
	If $(\mathcal{A}-\mathcal{N})$ is invertible, it follows necessarily that $(Id_n \otimes (\mathcal{A}-\mathcal{ N})^*+Id_n\circ\mathcal{A}^*)$ is also invertible, otherwise it contradicts the fact that the solution of the harmonic Lyapunov equation is uniquely defined. This concludes the proof.
\end{proof}
We are now in position to state one of the main results of this paper. To this end, define for any given $m$ the $m-$truncated solution as 
\begin{equation}col({\bf \tilde P}_m):=-(Id_n \otimes (\mathcal{A}_m-\mathcal{ N}_m)^*+Id_n\circ\mathcal{A}^*_m))^{-1}col({\bf Q}|_m)\label{inf_sol_2}\end{equation} with $\mathcal{A}_m:=\mathcal{T}_m(A)$, $\mathcal{N}_m:=Id_n\otimes diag(j\omega k,\ |k|\leq m)$ and where $Id_n\circ\mathcal{A}$ is defined by \eqref{symB}. The components of the $m-$truncated matrix ${\bf Q}|_m$ are given by $$({\bf Q}|_m)_{ij}:=\mathcal{F}|_m(q_{ij}),\ i,j := 1,\cdots, n,$$ with $\mathcal{F}|_m(q_{ij})$ the $m-$truncation of $\mathcal{F}(q_{ij})$ obtained by suppressing all phasors of order $|k|>m$.

\begin{theorem}\label{approx_P_n} Assume that $A(t)\in L^\infty([0 \ T])$ and $(\mathcal{A}-\mathcal{N})$ is invertible. For any given $\epsilon>0$, there exists $m_0$ such that 
	for any $m\geq m_0$: $$\|col({\bf {P}-\tilde {{P}}}_m)\|_{\ell^2}<\epsilon$$
	where ${\bf P}$, given by \eqref{inf_sym}, is the solution of the infinite-dimensional problem.
	Moreover, $$\|\mathcal{P}-\tilde{\mathcal{P}}_m\|_{\ell^2}<\epsilon$$
	with $\mathcal{P}:=\mathcal{T}(P)$ and $\tilde{\mathcal{P}}_m:=\mathcal{T}(\tilde P_m)$.
\end{theorem}

\begin{proof} It is sufficient to prove the theorem for $n:= 1$. 
	In this case, $A(z)=A_{11}(z)$ and $A(z)'=A(z)$. The symbol equation \eqref{sym_syl} reduces to:
	$2A(z)P(z)+N(z)\cdot P(z)+Q(z)=0$.
	
	Now, observe that the $k-$th coefficient of $A(z)P(z)$ for $ k\in \mathbb{Z} $ is provided by \begin{equation*}
		(A(z)P(z))_k=\sum_{ r \in\mathbb{Z}}A_{k-r}P_r\end{equation*}
	while the $k-$th coefficient of $N(z)\cdot P(z)$ is given by: $$(N(z)\cdot P(z))_k=j\omega k P_k.$$
	Thus, the symbol Lyapunov equation $2A(z)P(z)+N(z)\cdot P(z)+Q(z)=0$ can be rewritten equivalently by means to its coefficients as the infinite-dimensional linear system:
	
	\begin{equation}(2\mathcal{A}-\mathcal{N})^*{\bf P}=-{\bf Q}\label{inf_sol}\end{equation}
	where ${\bf P}$ and ${\bf Q}$ are infinite vectors whose components are the coefficients of $P(z)$ and $Q(z)$ (or equivalently the phasors of their time counterpart $P(t)$ and $Q(t)$).
	If a $m-$truncation on Equation \eqref{inf_sol} is applied, we obtain:
	$$(2\mathcal{A}_m-\mathcal{N}_m)^*{\bf P}|_m=-{\bf Q}|_m+E_m^++E_m^-$$
	where the correcting term is given by (see \eqref{z2} in Proposition~\ref{product3})
	$$E_m^+:=-2\mathcal{H}_{(m,\infty)}(A^+){\bf P}|_m^-$$
	$$E_m^-:=-2J\mathcal{H}_{(m,\infty)}(A^-){\bf P}|_m^+$$
	with ${\bf P}|_m^+:=(P_{{m+1}},P_{{m+2}},\cdots)$ and ${\bf P}|_m^-:=(P_{{-m-1}},P_{{-m-2}},\cdots)$. As the operator $(\mathcal{A}-\mathcal{N})$ is invertible, the matrix $\Lambda$ in Equation \eqref{lti} is also invertible
	as well as $2\Lambda$. Consequently, the spectrum of $ (2\mathcal{A}-\mathcal{N})$ is given by $ \sigma(2\mathcal{A}-\mathcal{N})=\{2\lambda+j\omega k: k\in \mathbb{Z}\}$.
	Therefore, $(2\mathcal{A}-\mathcal{N})$ is invertible as well as $(2\mathcal{A}-\mathcal{N})^*$.
	Then, Corollary~\ref{specbound} implies that $(2\mathcal{A}_m-\mathcal{N}_m)^*$ is invertible for $m$ sufficiently large.
	Now, define ${\bf \tilde P}_m$ the solution of the $m-$truncated problem by the following relation: 
	\begin{equation*}{\bf \tilde P}_m:=-(2\mathcal{A}_m-\mathcal{N}_m)^{*-1}{\bf Q}|_m.\end{equation*}
	Therefore, we have:
	$${\bf \tilde P}_m-{\bf P}|_m=(2\mathcal{A}_m-\mathcal{N}_m)^{*-1}(E_m^++E_m^-)$$
	with a $\ell^2$-norm bounded by (see Theorem~\ref{borne}):
	\begin{align*}
		\|{\bf \tilde P}_m&-{\bf P}|_m\|_{\ell^2}\leq\|(2\mathcal{A}_m-\mathcal{N}_m)^{*-1}\|_{\ell^2}\|(E_m^++E_m^-)\|_{\ell^2}\\
		&\leq2\|A\|_{L^\infty}\|(2\mathcal{A}_m-\mathcal{N}_m)^{*-1}\|_{\ell^2}(\|{\bf P}|_m^-\|_{\ell^2}+ \|{\bf P}|_m^+\|_{\ell^2}).\end{align*}
	As $(2\mathcal{A}_m-\mathcal{N}_m)^{*-1}$ is uniformly bounded (see Corollary~\ref{specbound}), and as $\|{\bf P}|_m^-\|_{\ell^2}=\|{\bf P}|_m^+\|_{\ell^2}\rightarrow 0$ when $m\rightarrow +\infty$, we conclude that for a given $\epsilon>0 $, there exists $m_0$ such that for any $m\geq m_0$, $\|{\bf P}|_m- {\bf \tilde P}_m\|_{\ell^2}< \epsilon/2$ and $\|{\bf P}-{\bf P}|_m\|_{\ell^2}< \epsilon/2$ as the phasors ${\bf P}_k\rightarrow 0$ when $k\rightarrow +\infty$.
Finally, we obtain: $$\|{\bf P}-{\bf \tilde P}_m\|_{\ell^2}< \epsilon$$
	assuming $({\bf \tilde P}_m)_k:=0$ for $|k|>m$. To prove the last assertion, that is $\|\mathcal{P}- \tilde {\mathcal{P}}_m\|_{\ell^2}< \epsilon$, notice that for $n:=1$, $col({\bf P}- {\bf \tilde P}_m)={\bf P}- {\bf\tilde P}_m.$
	When $n\geq 1$, invoking similar steps as before yields $\|col({\bf P}- {\bf \tilde P}_m)\|_{\ell^2}< \epsilon$. The proof is completed invoking Proposition~\ref{fro}, for any $n$.
\end{proof}
\begin{remark} Using the symbolic equation to derive the approximate solution allows to obtain a significant reduction of the computational burden since the linear problem defined by \eqref{inf_sol_2} is of dimension $n(2m+1)$ while the one defined by \eqref{trunc_lyap_m1} is of dimension $n^2(2m+1)^2$.\end{remark}

The following Corollary is interesting from a practical point of view in order to determine an accurate solution to the infinite harmonic Lyapunov equation from \eqref{inf_sol_2}. Indeed, for a prescribed $\epsilon>0$, it is sufficient to increases $m$ in \eqref{inf_sol_2} until \eqref{sym_LYAP} is satisfied.

\begin{corollary}\label{accurate_m}
	For a given $\epsilon>0$, there exists $m_0$ such that for any $m\geq m_0$, the symbol $\tilde P_m(z)$ associated to
	${\bf \tilde P}_m$ satisfies:
	\begin{equation}\|A(z)' \tilde P_m(z)+\tilde P_m(z)A(z)+(\mathbbm{1}_n\otimes N(z))\cdot \tilde P_m(z)+Q(z)\|_{\ell^2}<\epsilon \label{sym_LYAP}\end{equation} \end{corollary}
\begin{proof}It is sufficient to provide the proof for $n:=1$. If we evaluate the symbol equation with $\tilde P_m(z)$, by construction of $\tilde P_m(z)$, the result is given by 
	$A(z)' \tilde P_m(z)+\tilde P_m(z)A(z)+(\mathbbm{1}_n\otimes N(z))\cdot \tilde P_m(z)+Q(z)=E_m(z)$
	where $E_m(z):=A(z)'\tilde P_m(z)-(A(z)'\tilde P_m(z))|_m+\tilde P_m(z)A(z)-(\tilde P_m(z)A(z))|_m$
	with $(A(z)\tilde P_m(z))|_m$ the $m-$truncation of the product $(A(z)\tilde P_m(z))$.
	When $n:=1$, as $A(z)=A'(z)$ and as the coefficients of $A(z)$ are complex scalar numbers, $E(z)$ reduces to 
	$$E_m(z)=2(A(z)\tilde P_m(z)-(A(z)\tilde P_m(z))|_m.$$
	Thus, the non-zero coefficients of $E_m(z)$ are of degree $k$ for $|k|>m$,
	and are given by the following equation (see Equation \eqref{z1}, Proposition \ref{product3}): 
	\begin{align*}
		E^-_m:&=2J_\infty\mathcal{H}_{(\infty,m)}(A^-){\bf \tilde P}_m\\
		E_m^+:&=2\mathcal{H}_{(\infty,m)}(A^+)J_m{\bf \tilde P}_m.\end{align*}
	Consider $\tilde m:=\frac{m}{2}$ (assuming $m$ is an even number) and split $J_\infty\mathcal{H}_{(\infty,m)}(A^-)$ as follow :
	$$J_\infty\mathcal{H}_{(\infty,m)}(A^-)= [M_{1} \ M_2]$$ where $M_{1}$ corresponds to the first $\tilde m$ columns of $J_\infty\mathcal{H}_{(\infty,m)}(A^-)$ and 
	$M_2$ to its complement. 
	Then, it follows that:
	$$\|E^-_m\|_{\ell^2}\leq2\|M_{1}{\bf\tilde P_{1}}\|_{\ell^2}+2\|M_{2}{\bf\tilde P_{2}}\|_{\ell^2}$$ where ${\bf \tilde P}_m=({\bf \tilde P_{1}}, {\bf \tilde P_{2}})$.
	With this partition, it can be observed that $$M_{2}=\left(\begin{array}{cccc}
		\vdots& \ddots& &\vdots\\
		\vdots & \ddots& \ddots & \vdots\\
		A_{-\tilde m-2} & \ddots & \ddots &A_{-2m-2}\\
		A_{-\tilde m-1} & A_{-\tilde m-2} & \cdots& A_{-2m-1}\end{array}\right).$$
	Therefore, the $\ell^2$ norm satisfies:
	$$\|M_2\|_{\ell^2}\leq \|\mathcal{H}(A^-_{\tilde m})\|_{\ell^2}$$
	where $A^-_{\tilde m}(z)$ is the $\tilde m$-shifted symbol $A^-_{\tilde m}(z):=\sum_{k=1}^{+\infty} A_{-\tilde m-k}z^{-k}$.
	
	Using Theorem~\ref{borne}, the following bounds can be established:
	\begin{align*}
		\|M_{1}{\bf\tilde P_{1}}\|_{\ell^2}&\leq 2\|A\|_{L^\infty}\|{\bf \tilde P_{1}}\|_{\ell^2}\\
		\|M_{2}{\bf\tilde P_{2}}\|_{\ell^2}&\leq 2\|\mathcal{H}(A^-_{\tilde m})\|_{\ell^2}\|{\bf \tilde P_{2}}\|_{\ell^2}.
	\end{align*} 
	Since $\|{\bf \tilde P}_m-{\bf P}\|_{\ell^2}\rightarrow 0$ when $m\rightarrow +\infty$ and since the phasors ${\bf P}_k$ of ${\bf P}$ vanishe when $k\rightarrow+\infty$, 
	it follows that $\|{\bf \tilde P_{1}}\|_{\ell^2}=(\sum_{r=-m}^{-m/2}|{\bf \tilde P}_{r}|^2)^{\frac{1}{2}} \rightarrow 0$ when $m\rightarrow +\infty$.
	On the other hand, we have $\|\mathcal{H}(A^-_{\tilde m})\|_{\ell^2} \rightarrow 0$ when $\tilde m\rightarrow +\infty$ since $\|A^-_{\tilde m}(z)\|_{\ell^2}\rightarrow 0$.
	Therefore, for a given $\epsilon>0$, there exists $m_0$ so that for $m\geq m_0$, $\|E^-_m\|_{\ell^2}\leq \epsilon/2$. With similar steps, this is also the case for $\|E^+_m\|_{\ell^2}\leq \epsilon/2$ and the conclusion follows.
	
\end{proof}

We illustrate the results of this section using a $1-$dimensional example where the $T-$periodic state matrix is given by $A(t)=-1-\cos(t)+2\sin(t)+\cos(2t)$ 
($T=2\pi$) so that the associated symbol $A(z)$ is given by 
$$A(z)=2z^{-2}+(-1+2j)z^{-1}-1+(-1-2j)z+2z^{2}$$ and $\mathcal{A}_m$ is banded.
Having fixed $Q_m=Id_m$, if we attempt to solve the truncated harmonic Lyapunov equation (see Theorem~\ref{finite_hle_m1}), the Toeplicity of the obtained solution $P_m$ is clearly defective as shown in Fig.~\ref{fig2} by evaluating $$\log_{10} |P_m(i,j)-P_m(i+1,j+1)|,$$ $i,j:=1,\cdots, 2m$. It can be observed that this defect is mainly located in the upper leftmost corner and in the lower rightmost corner, when $m$ is chosen sufficiently large.

\begin{figure}[h]\begin{center}
		\includegraphics[width=\linewidth,height=6cm]{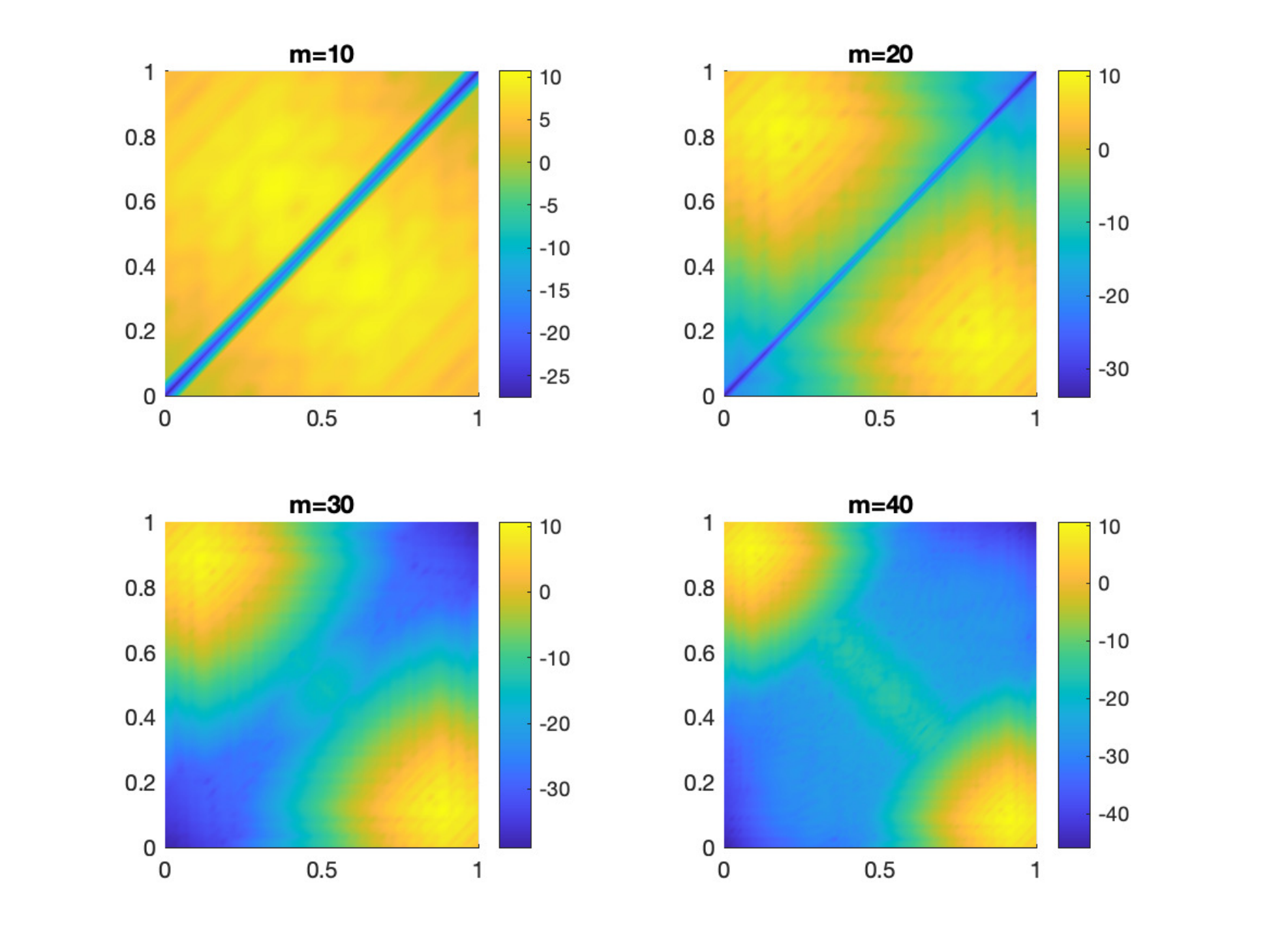}
		\caption{Plot of the "toeplicity" defect, $\log_{10} |P_m(i,j)-P_m(i+1,j+1)| $, $i,j:=1,\cdots ,2m$ of the Lyapunov solution $P_m$ w.r.t. the truncation order $m$. The axes are normalized.}\label{fig2}
	\end{center}
\end{figure}
We illustrate Corollary~\ref{accurate_m} on Fig.~\ref{fig4}. As expected, Fig.~\ref{fig4} shows that the absolute error produced by $\tilde P_m$ w.r.t. $m$ in the evaluation of the symbol Lyapunov equation decreases when $m$ increases. 
\begin{figure}[h]\begin{center}
		\includegraphics[width=\linewidth]{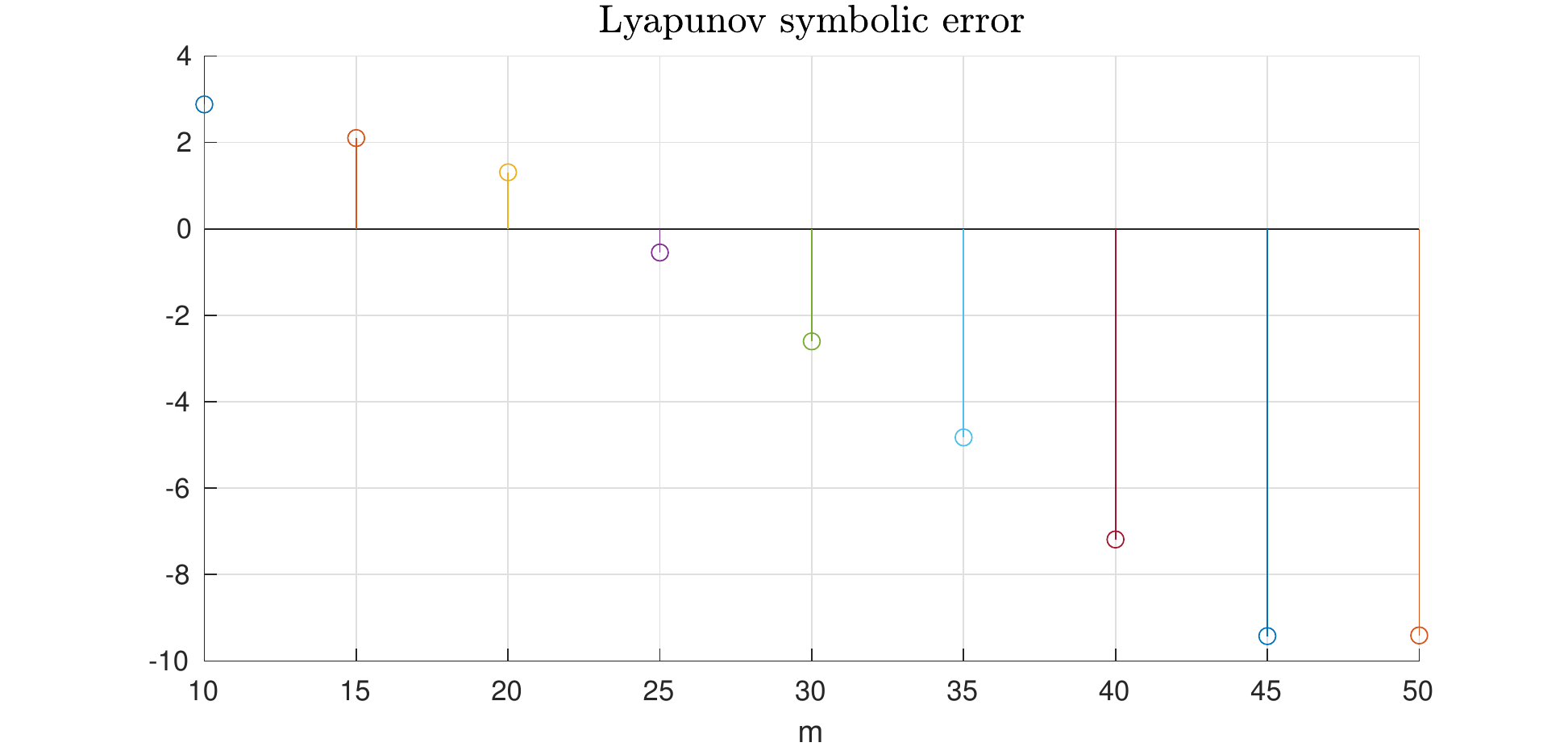}
		\caption{$\log_{10}\|2A(z)\tilde P_m(z)+N(z)\cdot \tilde P_m(z)+Q(z)\|$ w.r.t $m$}\label{fig4}
	\end{center}
\end{figure}
\begin{figure}[h]\begin{center}
		\includegraphics[width=\linewidth,height=3cm ]{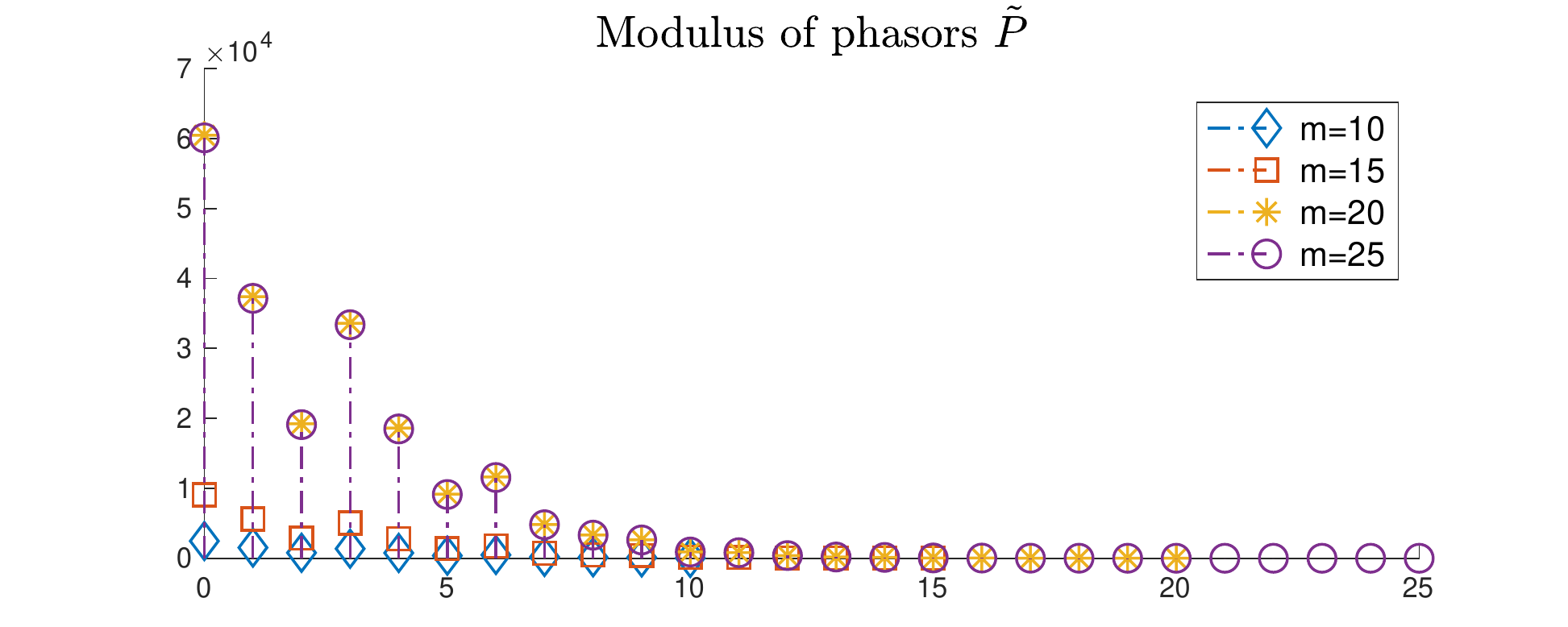}
		\caption{Phasors Modulus of $\tilde P$ w.r.t $m$}\label{fig44}
	\end{center}
\end{figure}
Fig.~\ref{fig44} shows the phasors modulus of $\tilde P$ obtained by \eqref{inf_sol_2} as a function of $m$. An accurate solution is clearly obtained when the modulus of the upper phasors vanishes i.e. here for $m\geq20$.

Finally, on Fig.~\ref{fig5} we plot $\log_{10} |\Delta P_m(i,j)|$, $i,j:= 1,\cdots,2m+1$ of the correction term $\Delta P_m$ (see Theorem~\ref{finite_hle_m1}). We observe that the support of $\Delta P_m$ is mainly located in the upper leftmost corner and in the lower rightmost corner, when $m$ is chosen sufficiently large. 
\begin{figure}[h]\begin{center}
		\includegraphics[width=\linewidth]{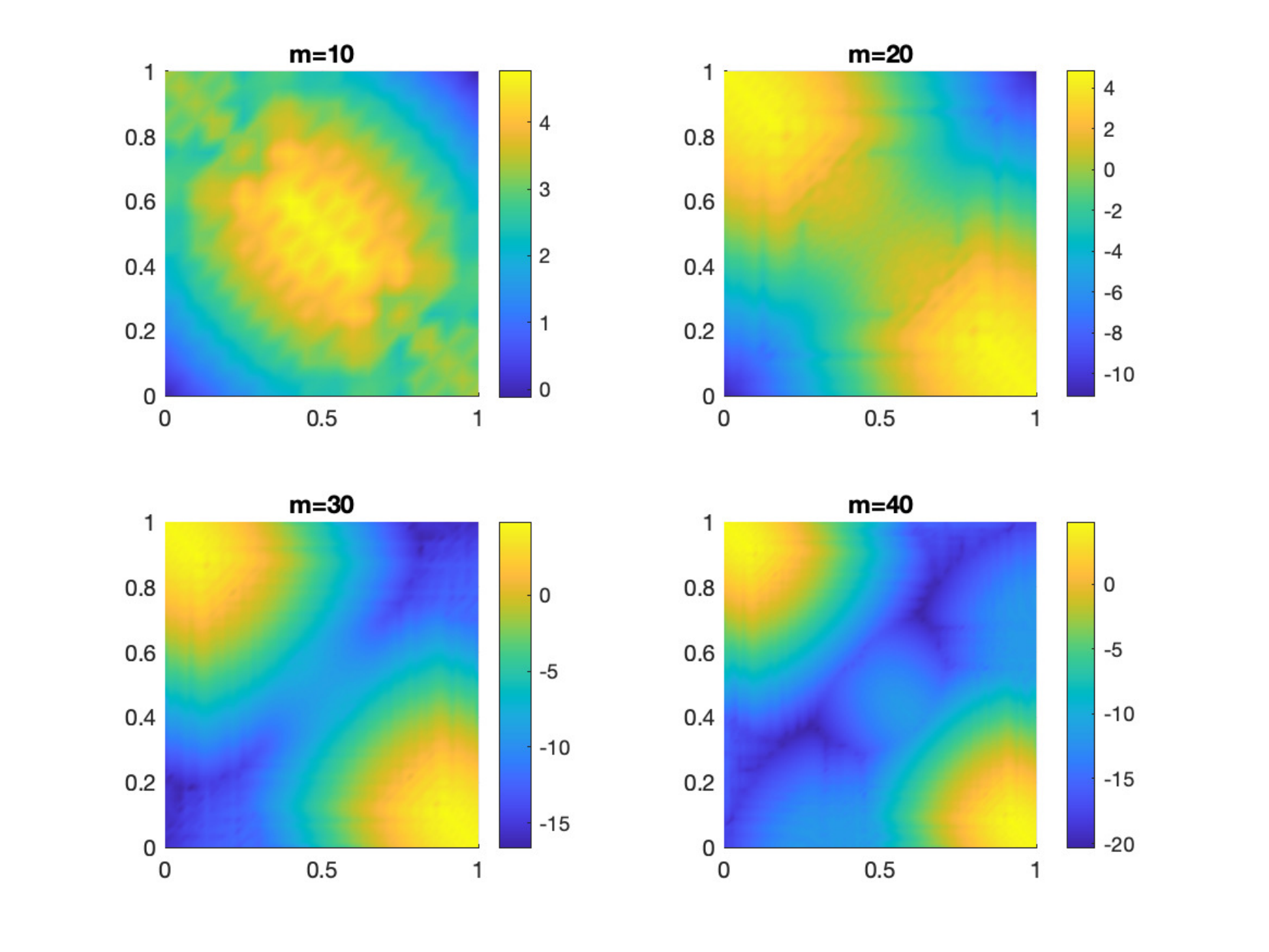}
		\caption{Plot of $\log_{10} |\Delta P_m(i,j)|$, $i,j:=1,\cdots,2m$ w.r.t. the $m-$truncation. The axes are normalized.}\label{fig5}
	\end{center}
\end{figure}

\section{Solving Harmonic Riccati Equations}
Here, we combine the proposed algorithm for solving harmonic Lyapunov equations and the Kleinman algorithm \cite{Kleinman} to solve harmonic Riccati equations. Recall that the Kleinman algorithm is a Newton based algorithm that allows to determine recursively the unique positive definite solution of a standard algebraic Riccati equation.
%
%

 Consider $T$-periodic symmetric positive definite matrix functions $R$ and $Q$ of $L^{\infty}$ class. Under the assumption that there exists $\eta>0$ such that the set $\{t: |det(R(t))|<\eta\}$ is of zero measure, it is proved in \cite{Blin} that $P$ is the unique $T$-periodic symmetric positive definite solution 
of the periodic Riccati differential equation:
\begin{align*}\dot P(t)&+A'(t)P(t)+P(t)A(t)\\&-P(t)B(t)R^{-1}(t)B(t)'P(t)+Q(t)=0,\end{align*}
if and only if the matrix $\mathcal{P}=\mathcal{T}(P)$ is the unique hermitian and positive definite solution of the algebraic Riccati equation:
\begin{equation}(\mathcal{A}-\mathcal{N})^*\mathcal{P}+\mathcal{P}(\mathcal{A}-\mathcal{N})-\mathcal{P}\mathcal{B}\mathcal{R}^{-1}\mathcal{B}^*\mathcal{P}+\mathcal{Q}=0,\label{are}\end{equation}
where $\mathcal{Q}=\mathcal{T}(Q)$ is hermitian positive definite. Moreover, $\mathcal{P}$ is a bounded operator on $\ell^2$.

Before generalizing the Kleinman algorithm to harmonic Riccati equations, we introduce the symbol Riccati equation and provide a link between a solution of a harmonic Riccati equation and a solution of the associated harmonic Lyapunov equation. 
\begin{proposition}$\mathcal{P}$ satisfies Equation \eqref{are} if and only if $P(z)$ satisfies the symbol Riccati equation
	\begin{align}A'(z)P(z)&+P(z)A(z)+(\mathbbm{1}_{n\times n}\otimes N(z))\cdot P(z)\nonumber\\
		&-P(z){B}(z)R^{-1}(z){B}'(z)P(z)+Q(z)=0 \label{sym_are}\end{align}
	where the operator $P(z)$ is bounded on $\ell^2$.
\end{proposition}
\begin{proof}The result is obtained using similar steps to those in the proof of Proposition~\ref{symLyap}.
\end{proof}

\begin{theorem}\label{link} Consider the harmonic Riccati equation \eqref{are} with $A(t)\in L^\infty([0 \ T])$. Let $\mathcal{K}:=-\mathcal{R}^{-1}\mathcal{B}^*\mathcal{P}$, $\mathcal{Y}:=\mathcal{K}^*\mathcal{RK}$ and ${\bf Y}:=\mathcal{F}(Y)$ with $Y:=\mathcal{T}^{-1}(\mathcal{Y})$. If $\mathcal{P}:=\mathcal{T}(P)$ is solution of \eqref{are} then ${\bf P}:=\mathcal{F}(P)$ satisfies:
	\begin{equation*}col({\bf P})=-\mathcal{M}^{-1}col({\bf Q+Y})\end{equation*}
	where $\mathcal{M}:=Id_n \otimes (\mathcal{A-BK}-\mathcal{ N})^*+ Id_n\circ (\mathcal{A-BK})^*$
	with $\mathcal{N}$ defined by \eqref{N} and ${\bf Q}:=\mathcal{F}(Q)$.
\end{theorem}
\begin{proof}
	If $\mathcal{P}:=\mathcal{T}(P)$ is the solution of \eqref{are} then it is the unique solution of the Lyapunov equation:
	\begin{equation}(\mathcal{A}-\mathcal{B}\mathcal{K}-\mathcal{N})^*\mathcal{P}+\mathcal{P}(\mathcal{A}-\mathcal{B}\mathcal{K}-\mathcal{N})+\mathcal{Q}+\mathcal{Y}=0.\label{syl}\end{equation}
	The result follows applying Theorem~\ref{sylvester_approx} to \eqref{syl}.
\end{proof}
In the next Theorem, we provide the algorithm to solve the infinite-dimensional harmonic Riccati equation \eqref{are} up to an arbitrary small error.
\begin{theorem}\label{t4} Assume that $A(t)\in L^\infty([0 \ T])$. For $k:=0,1,2,\cdots$ and for a sufficiently large $m_k:=m(k)$, define ${\bf \tilde S}_{m_k}(k)$ by:
	\begin{equation}col({\bf \tilde S}_{m_k}(k)):=-\mathcal{M}(k)^{-1}col({\bf Q}|_{m_k}+{\bf Y}(k)|_{m_k})\label{e2}\end{equation} 
	the $m_k-$truncated unique solution of the algebraic Lyapunov equation :
	\begin{equation}(\mathcal{A}(k)-\mathcal{N})^*\mathcal{S}(k)+\mathcal{S}(k)(\mathcal{A}(k)-\mathcal{N})+\mathcal{Y}(k)+\mathcal{Q}=0\label{syl2}\end{equation}
with ${\bf Y}(k):=\mathcal{F}(Y(k))$, $Y(k):=\mathcal{T}^{-1}(\mathcal{Y}(k))$ and $\mathcal{Y}(k)$ defined by
	\begin{align*}
		\mathcal{Y}(k)&:=\mathcal{K}(k)^*\mathcal{R}\mathcal{K}(k)\end{align*}
	and where $\mathcal{M}(k):=Id_n \otimes (\mathcal{A}_{m_k}(k)-\mathcal{ N}_{m_k})^*+Id_n\circ \mathcal{A}_{m_k}^*(k)$, $\mathcal{A}(k)$, its $m_k-$truncation $\mathcal{A}_{m_k}({k})$ and $\mathcal{Y}(k)$ are determined recursively by the symbols:
	\begin{align*}
		K_k(z)&:=R^{-1}(z)B(z)' \tilde S_{k-1}(z)\\
		A_k(z)&:=A(z)-B(z)K_k(z)\\
		Y_k(z)&:=K_k(z)'R(z)K_k(z)
	\end{align*}
	in which $\tilde S_{k-1}(z)$ denotes the symbol associated to	${\bf\tilde S}_{{m_{k-1}}}(k-1)$. Moreover, $K_0(z)$ is chosen such that the matrix $\mathcal{A}{(0)}-\mathcal{ N}:=\mathcal{A-B}\mathcal{K}_0-\mathcal{ N}$ is Hurwitz.\\
	Then, for $\epsilon>0$ sufficiently small, if $m_k$ is chosen sufficiently large at each step, we have:
	\begin{enumerate}
		\item $\|\mathcal{S}(k)-\tilde{\mathcal{S}}_{m_k}(k))\|_{\ell^2}<\epsilon$ where $\mathcal{S}(k)$ is the exact positive definite solution of \eqref{syl2} and $\tilde {\mathcal{S}}_{m_k}(k):=\mathcal{T}(\mathcal{F}^{-1}({\bf  \tilde S}_{m_k}(k)))$ with ${\bf  \tilde S}_{m_k}(k)$ given by \eqref{e2},
		\item $\mathcal{P}\leq \mathcal{S}(k+1)\leq \mathcal{S}(k)\leq\cdots\leq \mathcal{S}(0)$ where $\mathcal{P}$ solves \eqref{are}
		\item $\tilde{ \mathcal{S}}_{m_k}(k)>0$, for any $k=0,1,\cdots$, 
		\item $\lim_{k\rightarrow +\infty}\tilde {\mathcal{S}}_{m_k}({k}):=\tilde {\mathcal{S}}_{\bar m}$ with $\bar m=\lim_{k\rightarrow +\infty}m_k<+\infty$,
		\item $\|\mathcal{S}_{\infty}-\tilde {\mathcal{S}}_{\bar m}\|_{\ell^2}\leq \epsilon$
		where $\mathcal{S}_{\infty}:=\lim_{k\rightarrow +\infty}{\mathcal{S}}(k)$
		satisfies \eqref{are} with an error in $\ell^2$ norm given by
		\begin{equation}\|(\mathcal{A}-\mathcal{N})^*\mathcal{S}_{\infty}+\mathcal{S}_{\infty}(\mathcal{A}-\mathcal{N})-\mathcal{S}_{\infty}\mathcal{B}\mathcal{R}^{-1}\mathcal{B}^*\mathcal{S}_{\infty}+\mathcal{Q}\|_{\ell^2}\leq \eta \epsilon^2\label{areap}\end{equation}
		and $\eta:=\|{B}(t)\mathcal{R}^{-1}(t){B}'(t)\|_{L^\infty}$.\end{enumerate}
\end{theorem}
\begin{proof}
	We use Theorem~\ref{approx_P_n} in this proof as the related assumptions satisfied. For a given $\epsilon>0$ and for $k:=0$, we have from \eqref{e2}, $\tilde {\mathcal{S}}_{m_0}(0):=\mathcal{T}(\mathcal{F}^{-1}({\bf \tilde S}_{m_0}(0)))$ which differs from the exact solution $\mathcal{S}(0)$ of \eqref{syl2} by $$\|\mathcal{S}(0)-\tilde {\mathcal{S}}_{m_0}(0)\|_{\ell^2}\leq \epsilon$$ provided that $m_0$ is a sufficiently large number. 
	Thus, as $\mathcal{S}(0)$ is positive definite, so is $\tilde {\mathcal{S}}_{m_0}(0)$ provided that $\epsilon$ is small enough. 
	
	Recall that if $\mathcal{A-B}\mathcal{K-N}$ is Hurwitz, the solution of \eqref{syl2} is provided by $$\mathcal{S}_\mathcal{K}:=\int_0^{+\infty} {e^{(\mathcal{A-BK-N})^* \tau }}(\mathcal{Q}+\mathcal{K}^*\mathcal{R}\mathcal{K})e^{(\mathcal{A-BK-N})\tau}d\tau.$$
	
	Set $k:=1$ and consider $\mathcal{S}_{1}$ the bounded operator on $\ell^2$ solution of \eqref{syl2} obtained with $\mathcal{K}_1:=\mathcal{R}^{-1}\mathcal{B}^*\tilde {\mathcal{S}}_{m_0}(0)$ and $\mathcal{A}(1):=\mathcal{A-B}\mathcal{K}_1-\mathcal{N}$.
	Note that $\mathcal{S}_{1}$ is well defined since $\mathcal{A-B}\mathcal{K}_1$ and $\mathcal{Y}(1)+\mathcal{Q}$ are bounded operators on $\ell^2$ (or equivalently $\mathcal{T}^{-1}(\mathcal{A-B}\mathcal{K}_1)$ and $\mathcal{T}^{-1}(\mathcal{Y}(1)+\mathcal{Q})$ are $L^\infty$).
	Using similar steps as in the proof of \cite{Kleinman}, it can be established that:
	$$\mathcal{S}(0)- \mathcal{S}(1)=\int_0^{+\infty} e^{\mathcal{A}^*(1)\tau}(\mathcal{K}_0-\mathcal{K}_1)^*\mathcal{R}(\mathcal{K}_0-\mathcal{K}_1)e^{\mathcal{A}(1)\tau}d\tau\geq 0
	$$
	$$\mathcal{S}(1)-\mathcal{P}=\int_0^{+\infty} e^{\mathcal{A}^*(1)\tau}(\mathcal{K}_1-\mathcal{K})^*\mathcal{R}(\mathcal{K}_1-\mathcal{K})e^{\mathcal{A}(1)\tau}d\tau\geq 0
	$$
	where $\mathcal{P}$ is the solution of the Riccati equation \eqref{are}. 
	Therefore,
	$$0<\mathcal{P}\leq \mathcal{S}(1)\leq \mathcal{S}(0)$$
	which proves that $\mathcal{A}(1):=\mathcal{A-B}\mathcal{K}_1-\mathcal{N}$ is Hurwitz. 
	Using Theorem~\ref{approx_P_n}, the approximated solution $\tilde {\mathcal{S}}_{m_1}(1):=\mathcal{T}(\mathcal{F}^{-1}({\bf\tilde S}_{m_1}(1)))$ where ${\bf\tilde S}_{m_1}(1)$ is determined by \eqref{e2}, differs from $\mathcal{S}(1)$ by $\|\mathcal{S}(1)-\tilde {\mathcal{S}}_{m_1}(1))\|_{\ell^2}<\epsilon$ provided that $m_1$ is a sufficiently large number. Hence, $\tilde {\mathcal{S}}_{m_1}(1)$ is positive definite provided that $\epsilon$ is small enough.
	
	Repeating, for $k:=2,3,\cdots $ the above arguments, one gets:
	\begin{enumerate}
		\item $\mathcal{P}\leq \mathcal{S}(k+1)\leq \mathcal{S}(k)\leq\cdots\leq \mathcal{S}(0)$
		\item $\|\mathcal{S}(k)-\tilde {\mathcal{S}}_{m_k}(k)\|_{\ell^2}<\epsilon$
		\item for any $k$, $\mathcal{A}(k):=\mathcal{A-B}\mathcal{K}_k-\mathcal{N}$ is Hurwitz
	\end{enumerate}
	
	Recall that for any $k$, $\mathcal{S}(k)$ are bounded operators on $\ell^2$. Using monotonic convergence of positive operators, it follows that $\mathcal{S}_{\infty}:=\lim_{k\rightarrow +\infty}\mathcal{S}(k)$ exists with $\mathcal{S}_{\infty}$ a bounded operator on $\ell^2$ satisfying:
	\begin{equation}(\mathcal{A}-\mathcal{B}\tilde{\mathcal{K}}_{\infty}-\mathcal{N})^*\mathcal{S}_{\infty}+\mathcal{S}_{\infty}(\mathcal{A}-\mathcal{B}\tilde{\mathcal{K}}_{\infty}-\mathcal{N})+\tilde{\mathcal{K}}^*_{\infty}\mathcal{R}\tilde{\mathcal{K}}_{\infty}+\mathcal{Q}=0\label{sinf}\end{equation}
	where $\tilde{\mathcal{K}}_{\infty}:=\mathcal{R}^{-1}\mathcal{B}^*\tilde {\mathcal{S}}_{\bar m}$ with $\tilde {\mathcal{S}}_{\bar m}:=\lim_{k\rightarrow +\infty}\tilde {\mathcal{S}}_{m_k}(k)$ and $\bar m:=\lim_{k\rightarrow +\infty}m_k$.
	Therefore, $\mathcal{S}_{\infty}$ satisfies the following Riccati equation:
	\begin{align*}(\mathcal{A}-\mathcal{N})^*\mathcal{S}_{\infty}+\mathcal{S}_{\infty}(\mathcal{A}-&\mathcal{N})-\mathcal{K}_{\infty}^*\mathcal{R}\mathcal{K}_{\infty}+\mathcal{Q}=\nonumber\\
		&(\mathcal{K}_{\infty}-\tilde {\mathcal{K}}_{\infty})^*\mathcal{R}(\mathcal{K}_{\infty}-\tilde {\mathcal{K}}_{\infty})\end{align*}
	with ${\mathcal{K}}_{\infty}:=\mathcal{R}^{-1}\mathcal{B}^*\mathcal{S}_{\infty}$.
	
	As by construction $\|\mathcal{K}_{\infty}-\tilde {\mathcal{K}}_{\infty}\|_{\ell^2}\leq \zeta \|\mathcal{S}_{\infty}-\tilde {\mathcal{S}}_{\bar m}\|_{\ell^2}\leq \epsilon$ 
	where $\zeta:=\|\mathcal{R}^{-1}\mathcal{B}^*\|_{\ell^2}=\|{R}^{-1}{B}^*\|_{L^\infty}$ (see Theorem~\ref{borne}) and as $\mathcal{S}_{\infty}$ and $\mathcal{K}_{\infty}$ are bounded operators on $\ell^2$, we have necessarily that $\tilde {\mathcal{S}}_{\bar m}$ and $\tilde {\mathcal{K}}_{\infty}$ are also bounded on $\ell^2$. It follows that $\bar m$ is a finite number. 
	Indeed, as $\mathcal{S}_\infty$ solves \eqref{sinf} and as the assumptions of Theorem~\ref{approx_P_n} are satisfied, there exists a finite $m$ such that $\|\mathcal{S}_\infty-\tilde{\mathcal{S}}_m\|_{\ell^2} \leq \epsilon$. Consequently, $\bar m$ is finite.

	Now, taking the $\ell^2$- norm, we get:
	\begin{align*}\|(\mathcal{A}-\mathcal{N})^*\mathcal{S}_{\infty}+\mathcal{S}_{\infty}(\mathcal{A}-\mathcal{N})-&\mathcal{S}_{\infty}\mathcal{B}\mathcal{R}^{-1}\mathcal{B}^*\mathcal{S}_{\infty}+\mathcal{Q}\|_{\ell^2}\nonumber\\
		&\leq \eta\|\tilde {\mathcal{S}}_{\bar m}-\mathcal{S}_{\infty}\|^2_{\ell^2}\end{align*}
	where $\eta$ is such that $\|\mathcal{B}\mathcal{R}^{-1}\mathcal{B}^*\|_{\ell^2}=\|{B(t)}{R(t)}^{-1}{B(t)}^*\|_{L^\infty}<\eta$ (see Theorem~\ref{borne}).
	We have by construction $\|\tilde {\mathcal{S}}_{\bar m}-\mathcal{S}_{\infty}\|_{\ell^2}\leq \epsilon$ and the conclusion follows.
\end{proof}

This theorem shows that the algorithm returns a solution $\tilde {\mathcal{S}}_{\bar m}$ that approximates in $\ell^2$-norm operator sense the solution of the algebraic harmonic Riccati equation \eqref{are} and this approximation is characterized by \eqref{areap}.

\begin{remark} The choice of $m_k$ at each step $k$ must be sufficiently large to guarantee that $\|\mathcal{S}(k)-\tilde {\mathcal{S}}_{m_k}(k)\|_{\ell^2}<\epsilon$. This can be achieved by checking at each step a similar condition to the one provided in Corollary~\ref{accurate_m} using the symbol equation \eqref{sym_are}.
Moreover, the algorithm requires an initial step where the initial gain must be chosen such that $\mathcal{A}_{0}-\mathcal{ N}:=\mathcal{A-B}\mathcal{K}_0-\mathcal{ N}$ is Hurwitz. {This is not a major problem as one can use the pole placement procedure proposed in \cite{cdc2022} to design such a stabilizing $\mathcal{K}_0$.}
\end{remark}
 
\begin{remark}\label{init} Compared to \cite{Zhou2008} where an algorithm based on the iterative solution of the Lyapunov equation is proposed to solve harmonic Riccati equations, the algorithm of Theorem~\ref{t4} is more general as the matrices $A(t)$ and $B(t)$ belongs to $L^\infty$. Moreover, it is  assumed in \cite{Zhou2008} that $\mathcal{A}-\mathcal{N}$ is Hurwitz which is not the case here. Our algorithm applies to unstable harmonic matrices $\mathcal{A}-\mathcal{N}$ and it is also numerically more efficient with a significant reduction of the computational burden due to Theorem~\ref{approx_P_n}.
\end{remark}
\color{black}

\section{Harmonic LQ control design}
Consider a LTP system defined by
\begin{align}
	\dot x=&\left(\begin{array}{cc}a_{11} (t) & a_{12} (t) \\a_{21} (t) & a_{22} (t)\end{array}\right)x+\left(\begin{array}{c}b_{11}(t) \\0\end{array}\right)u\label{ex_ltp}\end{align}
where \begin{align*}a_{11} (t) &:=1+\frac{4}{\pi}\sum_{k=0}^{\infty}\frac{1}{2k+1}\sin(\omega (2k+1)t),\\
	a_{12} (t) &:= 2+\frac{16}{\pi^2}\sum_{k=0}^{\infty}\frac{1}{(2k+1)^2}\cos(\omega (2k+1)t),\\
	a_{21} (t) &:= -1+\frac{2}{\pi}\sum_{k=1}^{\infty}\frac{(-1)^k}{k}\sin(\omega kt+\frac{\pi}{4}),\\
	a_{22} (t) &:= 1-2\sin(2\pi t)-2\sin(6\pi t)+2\cos(6\pi t)+2\cos(10\pi t),\\
	b_{11}(t)&:=1+ 2 \cos(2\omega t)+ 4 \sin(6\omega t) \text{ with }\omega:=2\pi.
\end{align*}
Observe that $a_{11}$, $a_{12}$ and $a_{21}$ are respectively square, triangular and sawtooth signals and include an offset part. The associated Toeplitz matrix has an infinite number of phasors and is not banded. 
Moreover, this LTP system is unstable. The eigenvalues set is characterized by $\Lambda=\{1\pm j 1.64\}$. 
Let $\mathcal{Q}:=\mathcal{T}(100Id_n)$ and $\mathcal{R}:=\mathcal{T}(Id_m)$. We want to solve the associated Harmonic Riccati Equation using a $m-$truncation. We perform the study with $m \in \{ 8,16,32,64\}$. 

When one attempts to solve the $m-$truncated version of \eqref{are}, the Toeplicity defect 
for the solution $P_m$ and the associated gain $K_m$ is shown on Fig. \ref{fig6} and \ref{fig7}. 
\begin{figure}[h]\begin{center}
		\includegraphics[width=\linewidth]{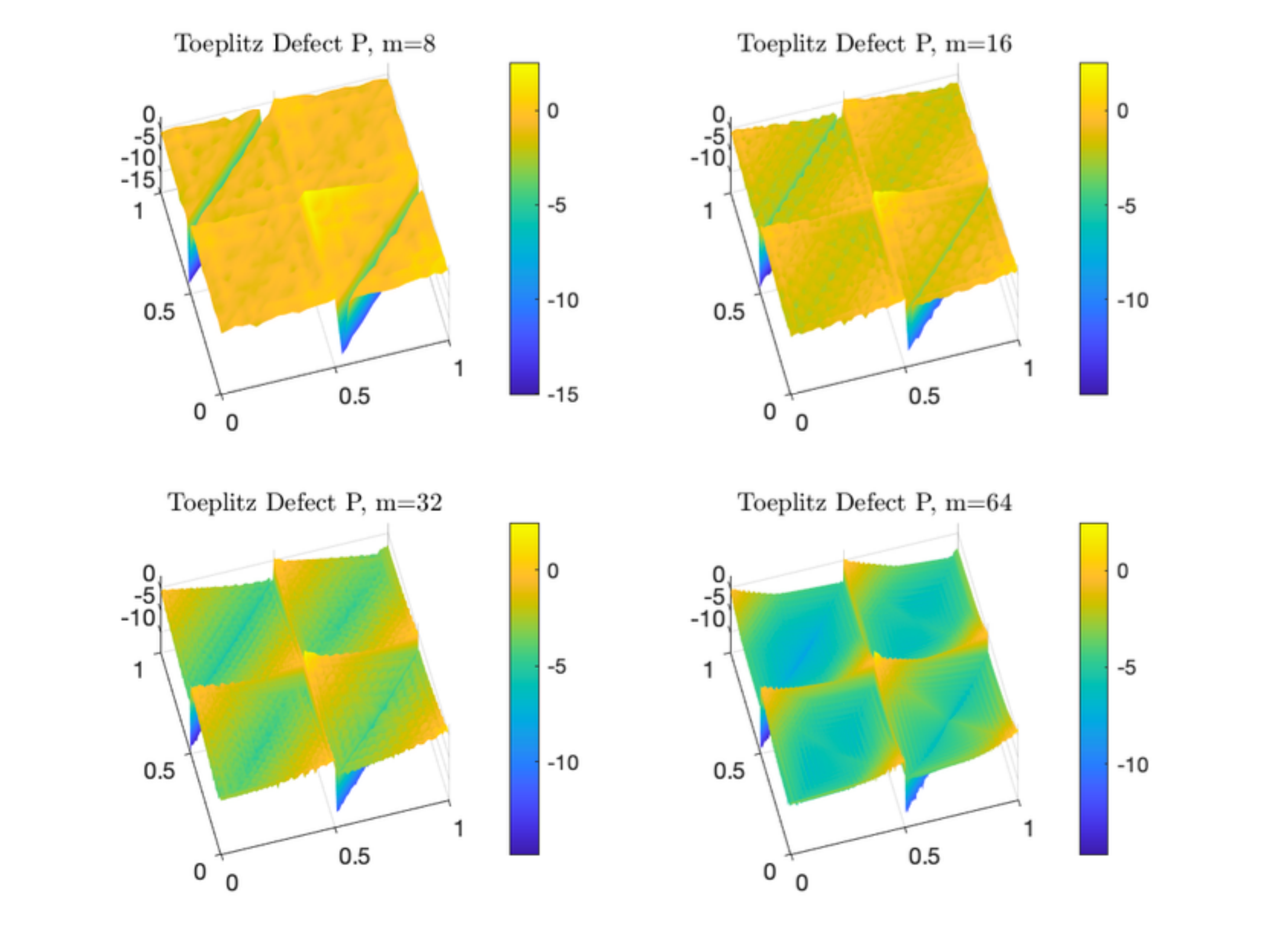}
		\caption{Plot (normalized axes) of the "toeplicity" defect, $\log_{10} |P_m(i,j)-P_m(i+1,j+1)| $, $i,j:=1,\cdots ,2m$ of the Riccati solution $P_m$ w.r.t. the truncation order $m$.}\label{fig6}
	\end{center}
\end{figure}
We observe that, for $m$ sufficiently large, this defect is mainly located at the upper left corner and lower right corner for both $P_m$ and $K_m$. As expected, this defect does not disappear when $m$ is increased. 
\begin{figure}[h]\begin{center}
		\includegraphics[width=\linewidth]{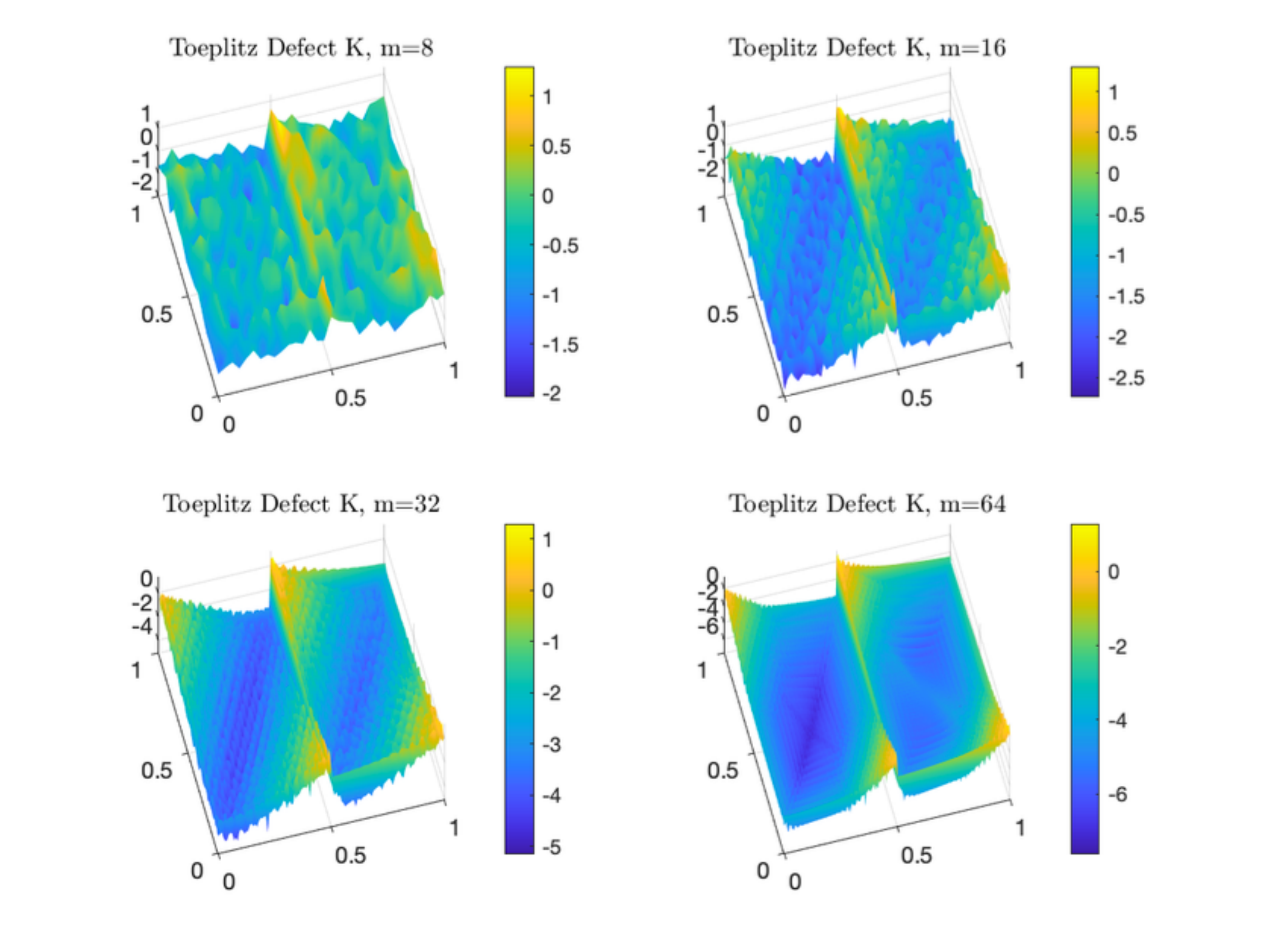}
		\caption{Plot (normalized axes) of the "toeplicity" defect, $\log_{10} |K_m(i,j)-K_m(i+1,j+1)| $, $i,j:=1,\cdots ,2m$ of the gain matrix $K_m=[K_1\ K_2]$ w.r.t. the truncation order $m$. }\label{fig7}
	\end{center}
\end{figure}
Now, we apply our algorithm to obtain an approximation of the infinite-dimensional Toeplitz solution. The correcting term is shown on Fig. \ref{fig8} and \ref{fig9}. When $m$ is chosen sufficiently large, we see that the correction terms both for $P_m$ and $K_m$ are mainly located at upper left and lower right corners of the corresponding $n \times n$ and $m \times n$ blocks.
\begin{figure}[h]\begin{center}
		\includegraphics[width=\linewidth]{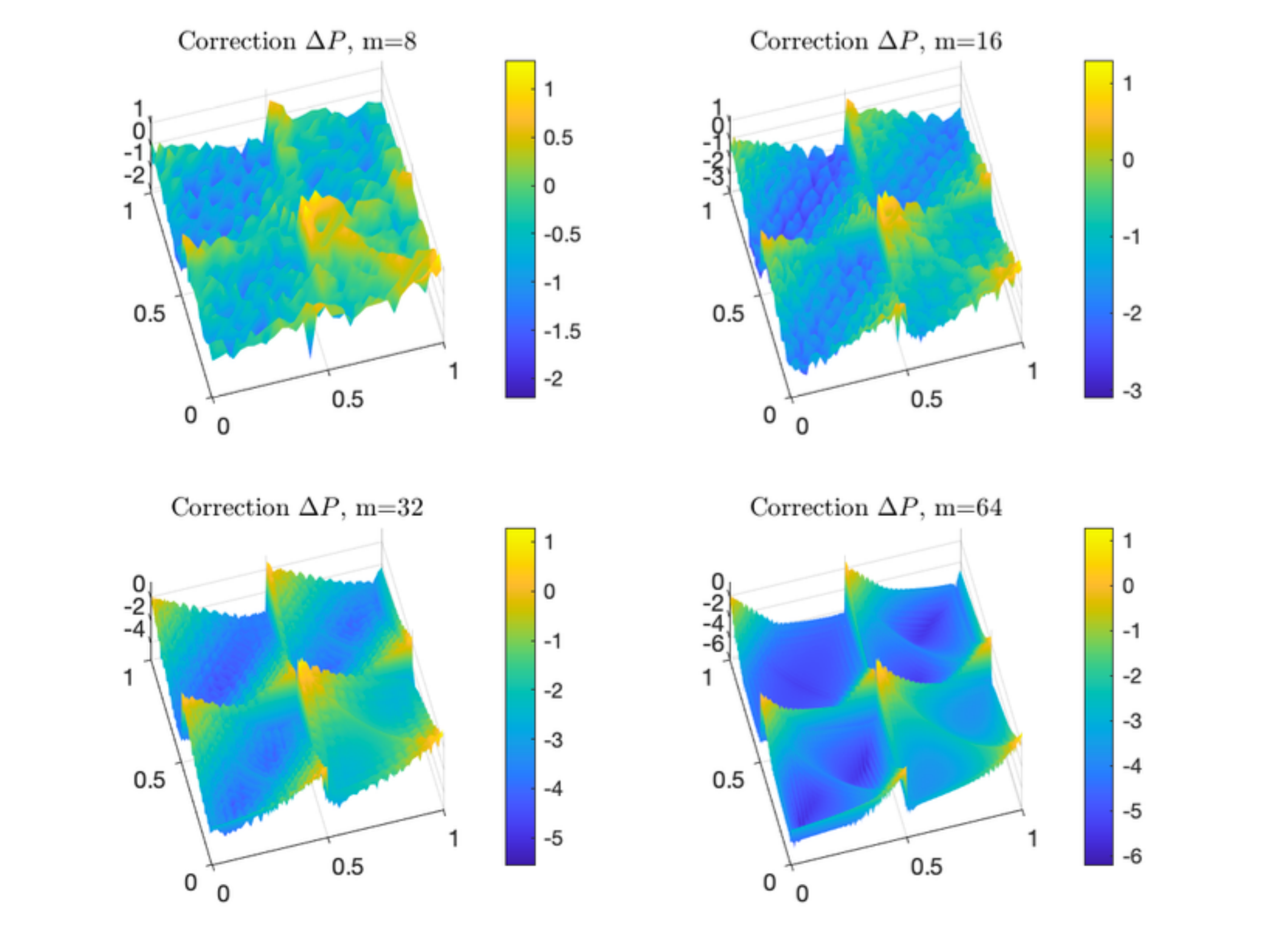}
		\caption{Plot (normalized axes) of $\log_{10} |\Delta P_m(i,j)|$, $i,j:=1,\cdots,2m$ w.r.t. the truncation order $m$.}\label{fig8}
	\end{center}
\end{figure}
\begin{figure}[h]\begin{center}
		\includegraphics[width=\linewidth]{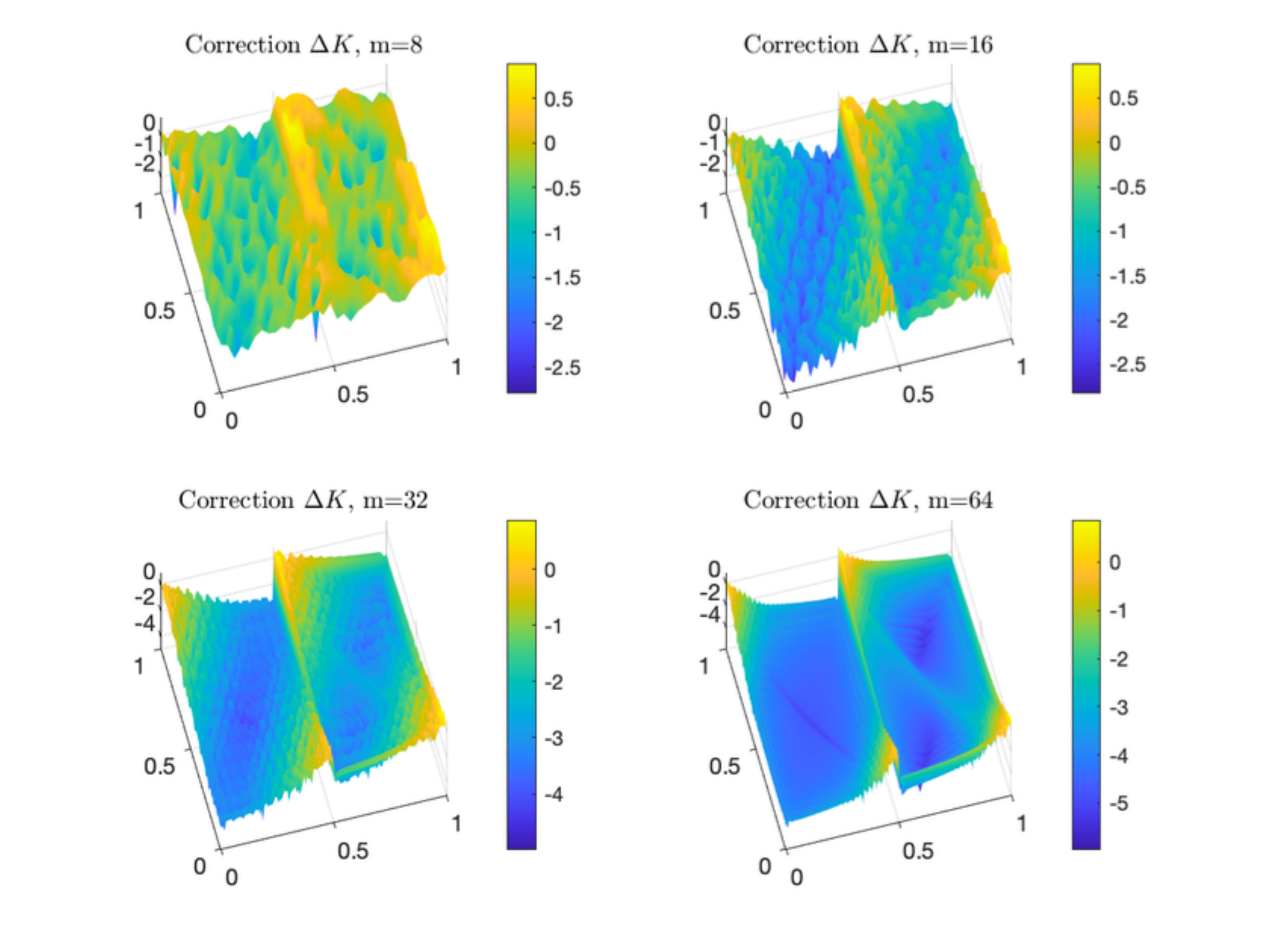}
		\caption{Plot (normalized axes) of $\log_{10} |\Delta K_m(i,j)|$, $i,j:=1,\cdots,2m$ w.r.t. the truncation order $m$.}\label{fig9}
	\end{center}
\end{figure}
\begin{figure}[h]
	\begin{center}
		\includegraphics[width=\linewidth]{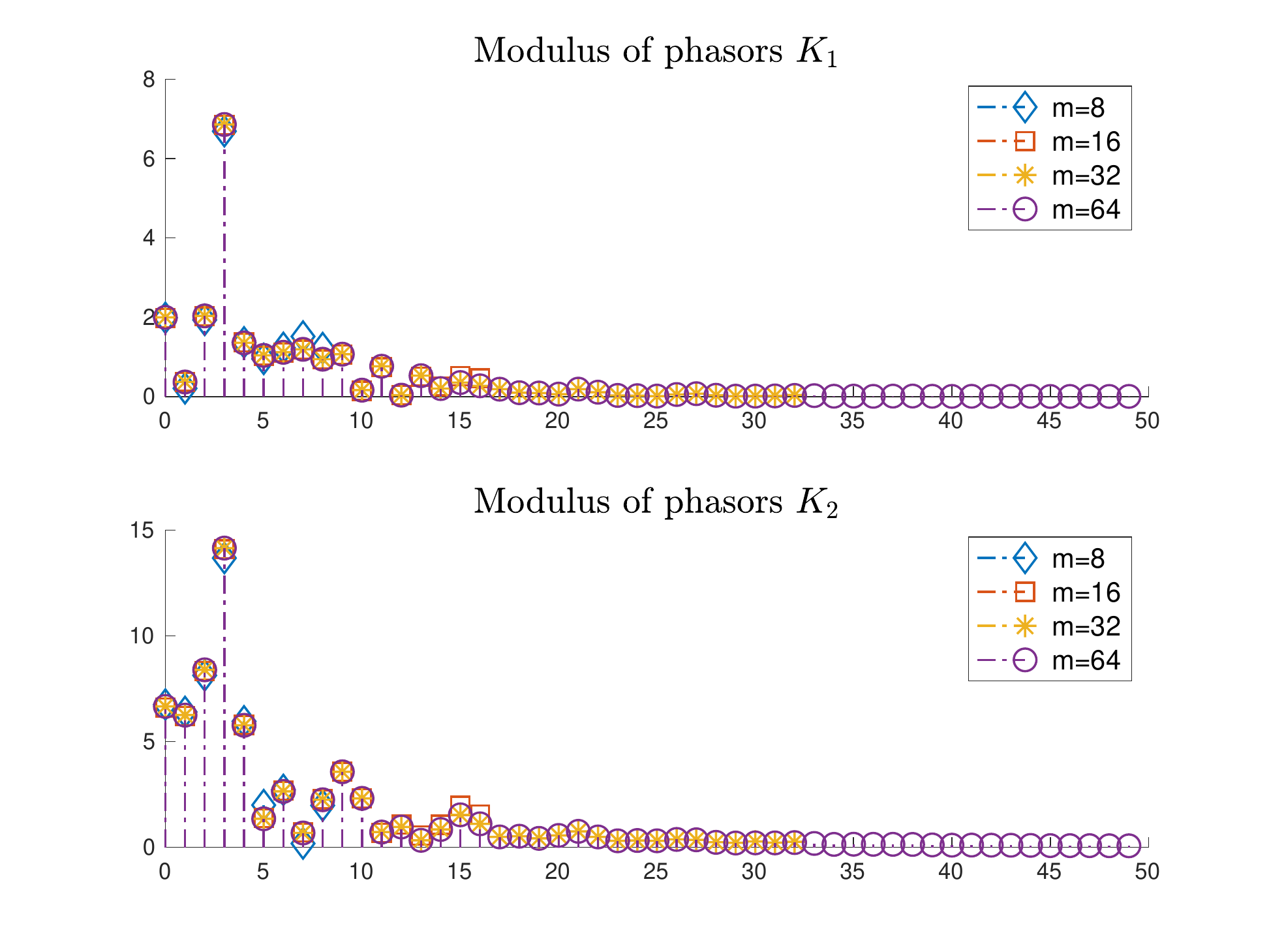}
		\caption{Modulus of Phasors $K=[K_1,K_2]$ w.r.t. $m:=8,16,32,64$ }\label{fig11}
	\end{center}
\end{figure}
\begin{figure}[h]
	\begin{center}
		\includegraphics[width=\linewidth]{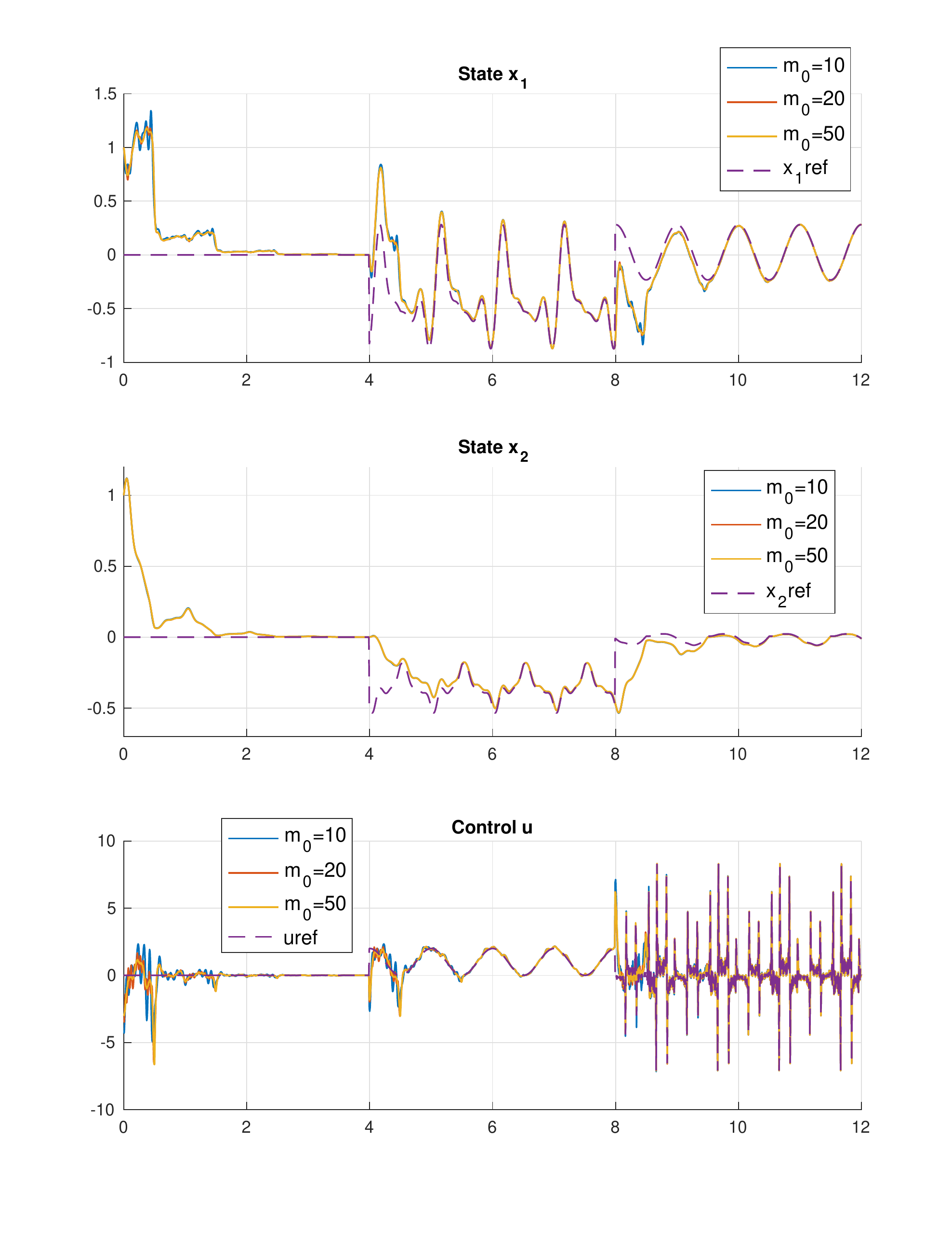}
		\caption{Closed loop response with $u(t):=-K(t)(x(t)-x_{ref}(t))+u_{ref}(t)$ and for $K(t)=\sum_{k=-{m_0}}^{m_0} K_k e^{j\omega kt}$ and $m_0:=10,20,50$. }\label{fig10}
	\end{center}
\end{figure}

 {Looking at the phasors of the harmonic gain matrix $\tilde{\mathcal{K}}_{\infty}:=\mathcal{R}^{-1}\mathcal{B}^*\tilde {\mathcal{S}}_{m}$ computed with a fixed $m_k:=m$ (we do not adapt $m_k$ at each step $k$ as described in Theorem~\ref{t4}) and plotted in Fig.~\ref{fig11} w.r.t. $m$, we see that the obtained values converge  and they vanish from $m\geq32$.}
The significant values are obtained for small values of $m$. To show the effectiveness of the proposed approach, we consider an equilibrium of the harmonic system defined by\begin{align}0=(\mathcal{A}-\mathcal{N})X_{ref}+\mathcal{B}U_{ref}\label{equi}\end{align}
and use (\ref{recos}) to reconstruct the associated $T-$periodic trajectory $x_{ref}=\mathcal{F}^{-1}(X_{ref})$ and control $u_{ref}=\mathcal{F}^{-1}(U_{ref})$. The control $$u(t):=-K(t)(x(t)-x_{ref}(t))+u_{ref}(t)$$ where $K(t)$ is the $T-$periodic gain matrix given by $K(t):=\sum_{k=-m}^m K_k e^{j\omega kt}$, stabilizes globally and asymptotically the unstable LTP system \eqref{ex_ltp} on any $T-$periodic trajectory $x_{ref}(t)=\mathcal{F}^{-1}(X_{ref})$ and $u_{ref}(t)=\mathcal{F}^{-1}(U_{ref})$. To illustrate this, we plot on Fig.~\ref{fig10} the closed loop response for three $T-$periodic reference trajectories. We start by $u_{ref}(t)=0$, for $t<4$, then $u_{ref}=1+\cos(2\pi t)$ for $4\leq t< 8$ and for $t\geq 8$ we consider a desired steady state $X_d$ given by $\mathcal{F}^{-1}(X_d)(t):=(\frac{1}{4}\cos(2\pi t),0)$ and look for the nearest harmonic equilibrium, solution of the minimization problem $\min_{U_{ref}}\|X_d-X_{ref}\|^2$ subject to \eqref{equi}. {Clearly it can be observed that the provided state feedback allows to track any $T-$periodic trajectory corresponding to any equilibrium of \eqref{equi}.} 

From a practical point of view, we see on this example that once a good approximation of the harmonic Riccati equation has been obtained for a sufficiently large $m$, a few number of coefficients $m_0\leq m$ are needed to reconstruct the matrix gain $K(t):=\sum_{k=-m_0}^{m_0} K_k e^{j\omega kt}$. This can be explained by the fact that the phasor gain modules vanish relatively quickly and can be approximated by $K_k\approx \mathcal{O}(\frac{1}{k})$ since $K(t)\in L^\infty$ (see Fig.~\ref{fig11}).

\section{Conclusion}
\color{black}
In this paper, a simple closed form formula for a Floquet factorization in the general case of $L^2$ matrix functions as well as a detailed spectrum characterization of harmonic state space operators and their $m-$truncations are provided. For any $m$, it has been proved that the spectrum of the truncated harmonic matrix may contain a  part that does not converge to the spectrum of the original infinite-dimensional harmonic matrix. From a practical point of view, these results have important consequences on the stability analysis when considering, for example, a $m-$truncation of the harmonic Lyapunov equation. We built upon this analysis efficient solutions to solve infinite-dimensional harmonic Lyapunov and Riccati equations up to an arbitrarily small error. These algorithms allow to recover the infinite-dimensional solution from a sequence of finite dimensional problems and the computational burden is reduced from  $n^2(2m+1)^2$ to $n(2m+1)$ where $n$ is the dimension of the LTP system and $m$ the order of the truncation. These results are important in practice and has been illustrated on the design of a harmonic LQ control with periodic trajectories tracking.

\bibliographystyle{plain}

\begin{thebibliography}{99} 
	\bibitem{Beam}Beam, R. and Warming, R., The asymptotic spectra of banded toeplitz and quasi-toeplitz matrices. SIAM Journal on Scientific Computing, pages 971-1006, July 1993.
	\bibitem{Bini}Bini, D. A., Massei, S., \& Meini, B. (2017). On functions of quasi-Toeplitz matrices. Sbornik: Mathematics, 208(11), 1628.
	\bibitem{Bini2}Bini, Dario A., Stefano Massei, and Leonardo Robol. "Quasi-Toeplitz matrix arithmetic: a MATLAB toolbox." Numerical Algorithms 81.2 (2019): 741-769.
	\bibitem{Blin}N. Blin, P. Riedinger, J. Daafouz, L. Grimaud-Salmon and P. Feyel, "Necessary and Sufficient Conditions for Harmonic Control in Continuous Time," in IEEE Transactions on Automatic Control, doi: 10.1109/TAC.2021.3117540.
	\bibitem{Bolzern}Bolzern, P., \& Colaneri, P. (1988). The periodic Lyapunov equation. SIAM Journal on Matrix Analysis and Applications, 9(4), 499-512.
	\bibitem{Bottcher}Bottcher, A., \& Grudsky, S. M. (2005). Spectral properties of banded Toeplitz matrices. Society for Industrial and Applied Mathematics.
	\bibitem{Castelli}Castelli, R., \& Lessard, J. P. (2013). Rigorous numerics in Floquet theory: computing stable and unstable bundles of periodic orbits. SIAM Journal on Applied Dynamical Systems, 12(1), 204-245.
	\bibitem{Farkas} Farkas, M.: "Periodic motions" (Springer-Verlag, New York, 1994)
	\bibitem{Iavernaro} Felice, I., Mazzia, F., and Trigiante, D., "Eigenvalues and quasi-eigenvalues of banded Toeplitz matrices: some properties and applications." Numerical Algorithms 31.1 (2002): 157-170.
	\bibitem{Floquet}
	Floquet, G. Sur les \'equations lin\'eaires a coefficients p\'eriodiques. 
	\emph{Annals Science Ecole Normale Sup\'erieure}, Ser. 2, 12, 47-88, 1883.
	\bibitem{Gohberg} 
	Gohberg, I., Goldberg, S. and Kaashoek, M.A., \emph{Classes of Linear Operators}, Operator Theory
	Advances and Applications Vol. 63 Birkhauser, Vol. II, 1993.
	\bibitem{gut} Guti\`errez-Guti\`errez, J., and Crespo, P. M. Block Toeplitz matrices: Asymptotic results and applications. Now, 2012.
	\bibitem{Kabamba}Kabamba, P. "Monodromy eigenvalue assignment in linear periodic systems." IEEE Transactions on Automatic Control 31.10 (1986): 950-952.
	\bibitem{Kleinman}Kleinman, D. "On an iterative technique for Riccati equation computations." IEEE Transactions on Automatic Control 13.1 (1968): 114-115.
	\bibitem{Massei}Massei, S., Palitta, D. and Robol, L., Solving Rank-Structured Sylvester and Lyapunov Equations, SIAM Journal on Matrix Analysis and Applications 2018 39:4, 1564-1590. 
	\bibitem{Montagnier}Montagnier, P., Spiteri, R. J., \& Angeles, J. (2004). The control of linear time-periodic systems using Floquet Lyapunov theory. International Journal of Control, 77(5), 472-490.
	\bibitem{Reichel92}Reichel, L., Trefethen, L.N., Linear Algebra and its Applications Volumes 162-164, February 1992, p.p. 153-185.
	\bibitem{cdc2022} P. Riedinger and J. Daafouz, Harmonic pole placement, submitted to IEEE CDC 2022, draft available on-line (Arxiv). 
	\bibitem{Robol}Robol, L., Rational Krylov and ADI iteration for infinite size quasi-Toeplitz matrix equations, R. L., Linear Algebra and its Applications, 2020, DOI: 10.1016/j.laa.2020.06.013.
	\bibitem{Sinha}Sinha, S. C., Pandiyan, R., \& Bibb, J. S. (1996). Liapunov-Floquet transformation: Computation and applications to periodic systems.(1996): 209-219.
	\bibitem{schmidt}Schmidt, P., \& Spitzer, F. (1960). The Toeplitz matrices of an arbitrary Laurent polynomial. Mathematica Scandinavica, 8(1), 15-38.
	\bibitem{Wereley_1990}Wereley, N. M., Analysis and control of linear periodically time-varying systems, \emph{Doctoral dissertation, Massachusetts Institute of Technology}, 1990.
	\bibitem{Zhou2004} Zhou, J., Hagiwara, T., \& Araki, M. (2004). Spectral characteristics and eigenvalues computation of the harmonic state operators in continuous-time periodic systems. Systems \& control letters, 53(2), 141-155.
	\bibitem{Zhou}Zhou, J. "Harmonic Lyapunov equations in continuous-time periodic systems: solutions and properties." IET Control Theory \& Applications 1.4 (2007): 946-954.
	\bibitem{Zhou2011}Zhou, B., \& Duan, G. R. (2011). Periodic Lyapunov equation based approaches to the stabilization of continuous-time periodic linear systems. IEEE Transactions on Automatic Control, 57(8), 2139-2146.
	\bibitem{Zhou2008}Zhou, J. Derivation and Solution of Harmonic Riccati Equations via Contraction Mapping Theorem, \emph{Transactions of the Society of Instrument and Control Engineers} 44(2), p.p. 156-163, 2008.
	\bibitem{Zhou2} Zhou, J. "Classification and characteristics of Floquet factorizations in linear continuous-time periodic systems." International Journal of Control 81.11 (2008): 1682-1698.
\end{thebibliography}

\begin{IEEEbiography}[{\includegraphics[width=1in,height=1.25in,clip,keepaspectratio]{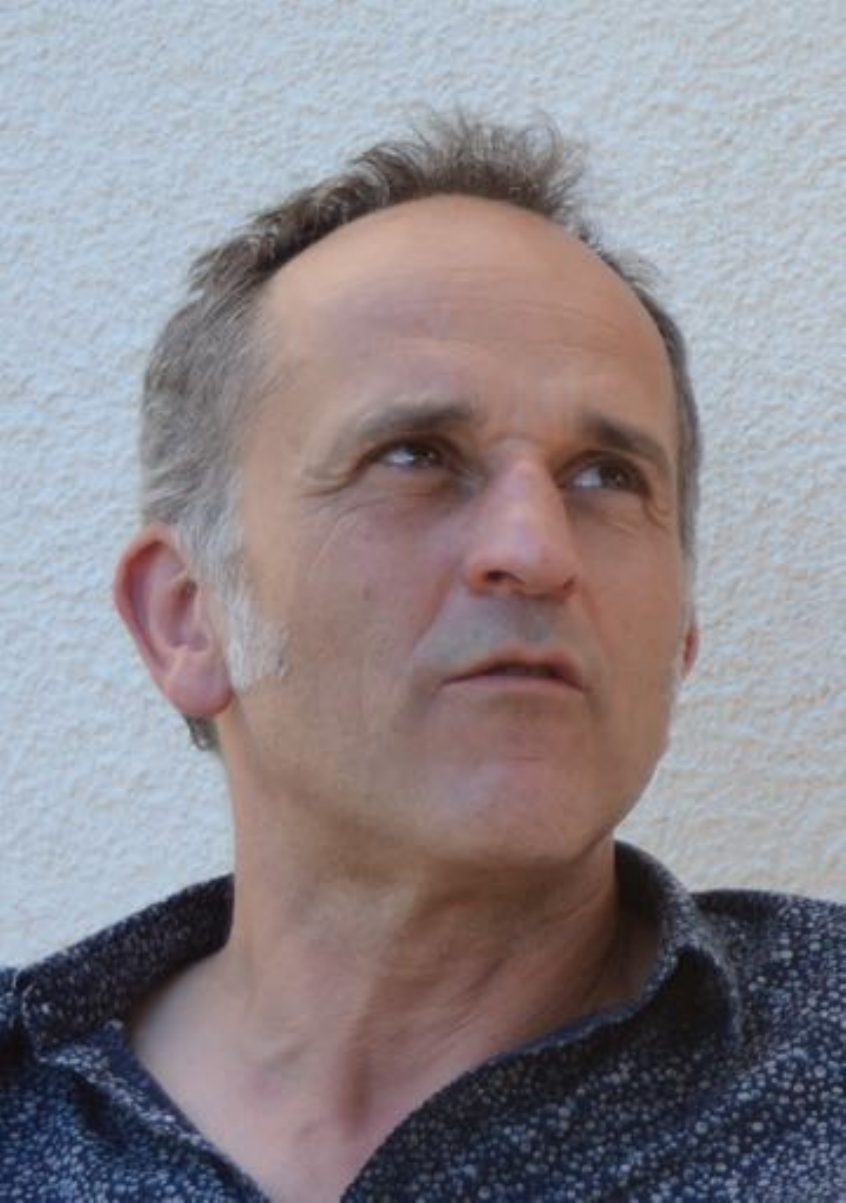}}]
	{Pierre Riedinger} is a Full Professor at the engineering school Ensem and researcher at CRAN - CNRS UMR 7039, Universit\'e de Lorraine (France).
	He received his M.Sc. degree in Applied Mathematics from the University Joseph Fourier, Grenoble in 1993 and the Ph.D. degree in Automatic Control 
	in 1999 from the Institut National Polytechnique de Lorraine (INPL). He got the French Habilitation degree from the INPL in 2010. His current research interests include control theory and optimization of 
	systems with their applications in electrical and power systems.
\end{IEEEbiography}
\begin{IEEEbiography}[{\includegraphics[width=1in,height=1.25in,clip,keepaspectratio]{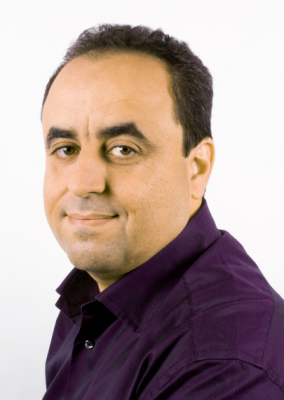}}]
	{Jamal Daafouz} 
	is a Full Professor at University
	de Lorraine (France) and researcher at CRAN-CNRS. In 1994, he received a Ph.D.
	in Automatic Control from INSA Toulouse, in 1997.
	He also received the "Habilitation à Diriger des
	Recherches" from INPL (University de Lorraine),
	Nancy, in 2005.
	His research interests include analysis, observation
	and control of uncertain systems, switched
	systems, hybrid systems, delay and networked systems with a particular
	interest for convex based optimisation methods.
	In 2010, Jamal Daafouz was appointed as a junior member of the
	Institut Universitaire de France (IUF). He served as an associate editor
	of the following journals: Automatica, IEEE Transactions on Automatic
	Control, European Journal of Control and Non linear Analysis and Hybrid
	Systems. He is senior editor of the journal IEEE Control Systems Letters. 
\end{IEEEbiography}
\end{document}